\documentclass{mathcryptology}
\usepackage[utf8]{inputenc}
\usepackage[T1]{fontenc}
\usepackage{amssymb}
\usepackage{algorithmic}
\usepackage{todonotes}

\ifpdf
  \DeclareGraphicsExtensions{.eps,.pdf,.png,.jpg}
\else
  \DeclareGraphicsExtensions{.eps}
\fi

\makeatletter
\def\@doletters#1{\def\@b{#1}\ifx\@b\@c\let\@a\egroup\else\@d{#1}\fi\@a}
\def\doletters#1#2{\bgroup\def\@d##1{\global\@namedef{#1##1}{\ensuremath{#2{##1}}}}\def\@a{\@doletters}\def\@c{;}\@a}
\makeatother
\doletters c\mathcal  ABCDEFGHIJKLMNOPQRSTUVWXYZ;
\doletters s\mathsf   ABCDEFGHIJKLMNOPQRSTUVWXYZ
                    {BW}{BL}{LD}{DL}{RSA}{NFS}{SNFS}{GNFS};
\doletters b\mathbb   ABCDEFGHIJKLMNOPQRSTUVWXYZ;
\doletters {fk}\mathfrak ABCDEFGHIJKLMNOPQRSTUVWXYZ
                         abcdefghijklmnopqrstuvwxyz;
\doletters {bf}\mathbf   ABCDEFGHIJKLMNOPQRSTUVWXYZ
                         abcdefghijklmnopqrstuvwxyz;

\DeclareMathOperator{\Z}{\mathbb{Z}}
\DeclareMathOperator{\R}{\mathbb R}
\DeclareMathOperator{\Res}{Res}

\DeclareMathOperator{\logTwo}{\log_2}
\DeclareMathOperator{\logThree}{\log_3}

\makeatletter
\let\olog\log

\def\olg{\log_2}
\def\mlog#1{\olog_{#1}}
\def\mlg#1{\olg^{(#1)}}

\newcount\@logcount
\def\@zlg{\ifnum\@logcount=1\olg\else
\edef\reserved@a{\noexpand\mlg{\the\@logcount}}\reserved@a
\fi}
\def\@lginner{\ifx\reserved@a\lg\def\next\lg{\@lg}\else\let\next\@zlg\fi\next}
\def\@lg{\advance\@logcount1\futurelet\reserved@a\@lginner}
\def\@zlog{\ifnum\@logcount=1\olog\else
    \edef\reserved@a{\noexpand\mlog{\the\@logcount}}\reserved@a
\fi}
\def\@loginner{\ifx\reserved@a\log\def\next\log{\@log}\else\let\next\@zlog\fi\next}
\def\@log{\advance\@logcount1\futurelet\reserved@a\@loginner}
\DeclareRobustCommand{\lg}{\@logcount0\@lg}
\DeclareRobustCommand{\log}{\@logcount0\@log}
\makeatother

\def\dbnorm#1{\mathopen{\mid\mskip-2mu\mid}#1\mathclose{\mid\mskip-2mu\mid}}

\newif\iffinal
\finalfalse

\iffinal
\renewcommand{\todo}[1][]{}
\fi

\date{}

\newtheorem{lemma}[theorem]{Lemma}
\newtheorem{proposition}[theorem]{Proposition}
\newtheorem{corollary}[theorem]{Corollary}
\newtheorem{definition}[theorem]{Definition}
\newtheorem{remark}[theorem]{Remark}
\newtheorem{example}[theorem]{Example}

\newtheorem{heuristic}[theorem]{Heuristic assumption}
\newtheorem{myproblem}[theorem]{Problem}
\newtheorem{conjecture}[theorem]{Conjecture}

\volume{1}
\issue{1}

\startpage{1}

\receiveddate{{\today}}
\revisiondate{{\today}}
\accepteddate{{\today}}

\articletitle{Refined Analysis of the Asymptotic Complexity of the Number Field Sieve}

\articleauthors{Aude Le~Gluher\aff{1,*}, Pierre-Jean Spaenlehauer\aff{1}, Emmanuel Thomé\aff{1}}

\articleaffiliations{
	\aff{1} Université de Lorraine, CNRS, Inria, LORIA, F-54000 Nancy, France
}

\correspondingauthoremail{aude.le-gluher@loria.fr} 

\citationauthors{Le~Gluher, A. \& Spaenlehauer, P-J \& Thomé E.} 


\MSClassification{68W40 (Analysis of algorithms), 11Y05 (Factorization).} 

\abstract{
	The classical heuristic complexity of the Number Field Sieve (NFS) involves an
  unknown function, usually noted $o(1)$ and called $\xi(N)$ throughout
  this paper, which tends to zero as the entry $N$ grows. The aim of this paper
  is to find optimal asymptotic choices of the parameters of
  NFS as $N$ grows, in order to minimize its heuristic asymptotic computational cost.
  This amounts to minimizing a function of the
  parameters of NFS bound together by a non-linear constraint.
  We provide precise asymptotic estimates of the minimizers of this
  optimization problem, which yield
  refined formulas for the asymptotic
  complexity of NFS. One of the main outcomes of this analysis is that $\xi(N)$
  has a very slow rate of convergence: We prove that it is equivalent to
  $4{\log}{\log}{\log}\,N/(3{\log}{\log}\,N)$. Moreover, $\xi(N)$ has an
  unpredictable behavior for practical estimates of the complexity. Indeed, 
  we provide an asymptotic series expansion of $\xi$ and numerical experiments indicate that this
  series starts converging only for $N>\exp(\exp(25))$, far beyond the practical range of NFS. This raises
  doubts on the relevance of NFS running time estimates that are based on
    setting $\xi=0$ in the asymptotic formula.
}

\begin{document}
\begin{NoHyper}
\articleinformation 
\end{NoHyper}

\section{Introduction}
Factoring integers and solving discrete logarithms in finite fields are
two fundamental problems in computational number theory, which are core routines
for a very large range of applications. Perhaps their most prominent and
critical use is the fact that the security of many currently deployed
cryptosystems --- e.g. RSA, finite field Diffie-Hellman, ElGamal --- relies directly on their computational
difficulty. This explains why the development of algorithms for solving these
two problems and the analysis of their complexity are central topics in
computational mathematics.

To this day, the Number Field Sieve (NFS) is the most efficient method to
factor integers and to compute discrete logarithms in prime finite fields.
It is an active research area, and implementations of NFS have led to several
record computations which give a good idea of the computational power
required to factor RSA moduli up to $2^{800}$, or to compute discrete
logarithms in prime fields of the same size~\cite{C:BGGHTZ20}. Existing software can
also be used to form reasonable estimates of the hardness of these
problems, 
up to roughly 1024 bits. 

Another approach to estimate
the
computational power needed is to use a theoretical cost analysis of NFS. 
The asymptotic complexity of the usual variant of NFS to factor an integer $N$, under various heuristic
assumptions, is known to be
\begin{equation}
  \label{eq:nfs-cplx-1}
  \exp\left( \sqrt[3]{\frac{64}{9}}(\log N)^{1/3}(\log\relax\log N)^{2/3}(1+\xi(N))\right)
\end{equation}
 where $\xi(N) \in o(1)$ as $N$ grows. This asymptotic formula is obtained by
 solving an optimization problem, which involves a number-theoretic function
 related to the notion of smoothness:

\begin{definition}
    \label{def:psi}
    Let $x\geq1$ and $y\geq1$ be two real numbers. An integer in $[1,x]$
    is \emph{$y$-smooth}
    if all of its prime factors are below $y$. We let
      \[\Psi(x,y) = \# \{ \text{integers in $[1,x]$ that are $y$-smooth}
        \}.\]
    In particular, $\Psi(x,y)/x$ is the probability that a uniformly
    random integer in $[1,x]$ is $y$-smooth. For convenience, we also use
    the notation $
        p(u,v) = \log\left(\Psi(e^u,e^v)/e^u\right)$.
\end{definition}

Unfortunately, the complexity given by Formula~\eqref{eq:nfs-cplx-1} is
not very satisfactory. First, it relies on several heuristics, most notably on
the fact that norms of random ideals in a given
number field have the same smoothness probability as random integers of the
same order of magnitude. A second point is that even if we consider that these assumptions hold, the
complexity given by Formula~\eqref{eq:nfs-cplx-1} involves a function $\xi$ which is never spelled out explicitly.

This inaccuracy conflicts with the fact that for the last few decades, the widespread development
of public-key cryptography has created a pressing need to answer the
following question:
Given
some computational power $C$, what size of RSA modulus $N$ should be considered so
that the cost of NFS exceeds $C$?
This is of interest, for example, to
regulatory bodies, see for
example~\cite[Sec.~7.5]{CSRC},
\cite[Sec.~B.2.2]{anssi-rgs}, or~\cite[Table 3.1]{enisa}.
A common way to answer this is to consider
the complexity formula, assume that $\xi = 0$ since we know $\xi(N) \in o(1)$ when $N$ grows, and use a computational record to
set a proportionality ratio. Since the 2000s, the validity of this
approach has been considered differently.
In~\cite[\S
2.4.4]{lenstraverheul2001}, the reader is warned
that $\xi$ should not be neglected for large extrapolations, and
omitting it probably yields biased
estimations.
Later work gradually moved to considering $\xi=0$ as a totally acceptable
simplification assumption, not heeding the warning. We emphasize that this method for computing key sizes is the
international standard for deployed RSA-based cryptography.

The goal of this paper is to question the relevance of neglecting $\xi$ and to give
insights on what this function hides. In particular,
assuming the standard heuristics for the complexity of NFS, we
show (Thm.~\ref{th:first_terms}) that 
\[\xi(N) = \frac{4{\log}{\log}{\log}\,
N}{3{\log}{\log}\,N} + \frac{-2\log 2+\log 3/6-2}{{\log}{\log}\, N}(1+o(1)).\]

This asymptotic
expansion of $\xi$ can be extended to a bivariate series expansion
evaluated at ${\log}{\log}{\log}\,N/{\log}{\log}\,N$ and $1/{\log}{\log}\,N$.
We provide algorithms --- together with publicly available implementations in
\textsf{SageMath} --- with
which we were able to compute more than a hundred of
terms of this asymptotic expansion. This may sound like
good news for practical formula-based keysize computations, as it seems that this
would provide more precise estimates. In fact, the contrary happens: We
observe experimentally that values of $N$ that would be useful for cryptography are
far too small compared to the radius of convergence at infinity of the series involved in the expansion of $\xi$. 
Said otherwise, this means that the exponent in the classical formula used for estimating
practical RSA key sizes is the first term of a divergent series. In order to illustrate
this phenomenon, let us consider the following example functions:
\[g_0: N\mapsto \exp\left({(\log
N)^{1/3}({\log}{\log}\,N)^{2/3}}\right);\quad
g: N\mapsto \exp\left(\frac{(\log
N)^{1/3}({\log}{\log}\,N)^{2/3}}{1+20/{\log}{\log}\,N}\right).\]

While it is true that $g$ belongs to the class $\exp((\log
N)^{1/3}({\log}{\log}\,N)^{2/3}(1+o(1)))$, the radius of convergence of
$1/(1+x)=\sum_i (-x)^i$ is $1$. This implies that if
$N\leq \exp(\exp(20))\approx 2^{699945421}$, then equalling $o(1)$ with $0$
 means estimating the exponent of the function $g$ by
evaluating the first term of
a divergent series. Numerical computations show that $g(2^{2048})\approx
2^{16}$, while $g_0(2^{2048})\approx
2^{61}$, so estimating $o(1)$ by $0$ in the case $N=2^{2048}$ yields
completely erroneous results. If one were to compare ``projected
values'' between $N=2^{512}$ and $N=2^{2048}$ based on the simplification
$o(1)=0$, the corresponding calculation would yield
$g_0(2^{2048})/g_0(2^{512})\approx 2^{28}$, compared to
$g(2^{2048})/g(2^{512})\approx 2^{9}$. Carelessly neglecting the $o(1)$
term can lead to dramatic errors in the assessment of a complexity in the
class given by Formula~\eqref{eq:nfs-cplx-1}, and of its growth as $N$ varies.

The asymptotic expansion of $\xi$ that we obtain in the complexity of NFS exhibits a behavior similar to the example function $g$. This raises questions on the relevance of using
the asymptotic formula for estimating practical cryptographic key lengths. In
particular, we obtain completely different results for the final NFS complexity depending on
how many terms we consider in our series expansion of $\xi$, see Figure~\ref{fig:zonecrypto}. This asks --- at
the very least --- for strong justification when one chooses to set
$\xi=0$ in the complexity formula.

Beyond NFS, the complexity of several algorithms is linked to the
asymptotics of smoothness probabilities. This is the case for example of
the quadratic sieve and its variants, of class group computations in
number fields, or of the elliptic curve factoring method. In all of these
cases, the methods of this article apply, and yield similar conclusions;
NFS is just an example.

\subsection*{Organization of the paper} In
Section~\ref{sec:NFS_background}, we briefly describe 
NFS and we state the optimization problem whose minimum is the complexity.
Section~\ref{sec:smoothness} provides a refined asymptotic expansion of
Dickman--De Bruijn function at infinity,
which is used to estimate the smoothness probabilities in NFS. Then
Section~\ref{sec:refined_compl} is devoted to the series expansion of the
function~$\xi$. Finally, Section~\ref{sec:expe} reports on
experimental results obtained by using the refined asymptotic formulas for the
complexity.

\subsection*{Implementation}
The computation of the asymptotic expansion of $\xi$ relies on three
algorithms described in Section~\ref{sec:coeffs_expansion}. Our
implementation in \textsf{SageMath} of these algorithms is available at
the following URL
\begin{center}
\url{https://gitlab.inria.fr/NFS_asymptotic_complexity/simulations}. 
\end{center}

At the same location, we also detail some of the
unilluminating technical calculations that are omitted for brevity in the text.

\subsection*{Notations and conventions}
\label{subsec:notations-log}
Throughout the article,
$\log x$ denotes the natural logarithm to base $e$. We use the notation $\mlog{m}$
to denote the $m$-th iterate of the $\log$ function, so that
$\mlog{m+1}=\log\circ\mlog m$.
The notation $u=\Theta(v)$ denotes: $(u=O(v)\text{ and
}v=O(u))$.
Throughout the paper, we assume a RAM computation model, where the cost of memory accesses is
neglected.

\section{Background on NFS}\label{sec:NFS_background}

In a nutshell, NFS defines two irreducible integer polynomials $f_0$
and $f_1$, and searches for integer pairs $(u,v)$ in a search space \cA\
such that the integers $\Res_x(u-vx,f_0(x))$ and $\Res_x(u-vx,f_1(x))$ are
smooth with respect to chosen smoothness bounds $B_0$ and $B_1$, where
$\Res_x$ denotes the resultant with respect to $x$. The
notations $K_i$ for $i=0,1$ denote the number fields $\bQ[x]/f_i(x)$.
The number of smooth pairs that must be collected must be at least 
$(\pi_{K_0}(B_0)+\pi_{K_1}(B_1))$, which is
the number of prime
ideals with norm at most $B_0$ and $B_1$ in the rings of integers of $K_0$ and $K_1$.
A linear algebra calculation
follows, and its dimension is again $(\pi_{K_0}(B_0)+\pi_{K_1}(B_1))$.
These two steps of the algorithm are the most costly.

In order to analyze the algorithm, we consider a simplified version,
which has of course little to do with computational feats that use the
Number Field Sieve~\cite{C:BGGHTZ20}.
In particular, we consider the straightforward ``base-$m$'' polynomial
selection. For a chosen degree $d$, we set
$m=\lceil N^{1/(d+1)}\rceil$, $f_0=x-m$, and $f_1$ of degree $d$ such
that $f_1(m)=N$ and $\dbnorm{f_1}_\infty<m$. We have $K_0=\bQ$,
and $K_1$ is a degree~$d$ number field. Side~0 is therefore called the
rational side, and side~1 is
called the
algebraic side. We assume that the smoothness bounds $B_0$ and $B_1$ can be
chosen equal to the same bound $B$ without increasing the overall asymptotic
complexity. Letting $B_0$ and $B_1$ have distinct values would add an extra
layer of technicality to the analysis, but we acknowledge that it would be
interesting to investigate the general case in future work.
In our simplified NFS algorithm, we pick the pairs $(u,v)$ from a set $\cA=[-A,A]^2$,
for some bound~$A$.
The integers that we check for smoothness in NFS are bounded by $M_0$ and
$M_1$, which we define as follows
        \begin{align*}
          \lvert\Res_x(u-vx,f_0(x))\rvert&\leq M_0 = (m+1)A,\\
            \llap{and }
            \lvert\Res_x(u-vx,f_1(x))\rvert&\leq M_1 = (d+1)mA^d.
        \end{align*}
        
It is important to point out the heuristic nature of the estimate given by
Formula~\eqref{eq:nfs-cplx-1}, which is related to the estimation of the
probabilities of smoothness.

\begin{heuristic}
    \label{heuristic}
    In the
    Number Field Sieve algorithm, the probability 
    that the two integers
    $\Res_x(u-vx, f_0)$ 
    and
    $\Res_x(u-vx, f_1)$ 
    are simultaneously $B$-smooth,
    as $(u,v)$ are picked uniformly from the search space \cA, is given by the
    probability that two random integers of
    the same size are $B$-smooth, which is $
    \frac{\Psi(M_0,B)}{M_0}
    \cdot
    \frac{\Psi(M_1,B)}{M_1}.$
\end{heuristic}

This assumption is very bold, and in fact wrong in several ways. 
Correcting terms, which are actually used in practice, can be used to
lessen the gap between $\Res_x(u-vx, f_0)$ and integers of the same
size~\cite{Murphy99,barbulescu-lachand}. 
However, this heuristic is not that wrong asymptotically, as
evidenced by~\cite[Eq.~(1.21)]{Granville}.
We also note that if the number field is a random variable, positive
results do exist~\cite{lee-venkatesan}.

\medskip
Our goal is not to get rid of this heuristic assumption. The usual
complexity analysis of NFS already uses it as a base and we will do the
same to improve on the asymptotic estimate in Formula~\eqref{eq:nfs-cplx-1}.

Our simplified version of NFS
has three main parameters that we can tune, as functions of
$\log N$, in order
to optimize the asymptotic complexity: the degree $d=\deg f_1$,
and the bounds $A$ and $B$. The time complexity of NFS is then
the sum of the time taken by the search for smooth pairs and the time
taken by linear algebra. We write these as
\begin{align*}
  C_{\text{search}} &= A^2\cdot C_{\text{test}},\\
    C_{\text{linear algebra}} &=
    (\pi_{K_0}(B)+\pi_{K_1}(B))^2(1+o(1)).
\end{align*}

In the expressions above,
        $C_{\text{test}}$ is the time spent per pair $(u,v)$.
Several techniques can be used to
bring $C_{\text{test}}$ to an amortized cost of $O(1)$; sieving is the
most popular, but ECM-testing also works (heuristically), and it is also
possible to detect smooth values with product trees as
in~\cite{djb-sf-2002}. Furthermore, we counted a quadratic cost for linear
algebra, as this can be carried out with sparse linear algebra
algorithms such as~\cite{Wiedemann86}.
We wish to minimize the cost, subject to
one constraint: We need enough smooth pairs, as we mentioned in the first
paragraph of this section.
\begin{equation*}
    \label{eq:constraint}
    A^2
    \cdot
    \frac{\Psi(M_0,B)}{M_0}
    \cdot
    \frac{\Psi(M_1,B)}{M_1}
\geq
    \pi_{K_0}(B)
  + \pi_{K_1}(B).
\end{equation*}

    Of course, we
can assume that this is an equality, since otherwise decreasing 
$A$ would decrease the complexity.

We are looking for expressions for $d$, $A$, and $B$
that are functions
of $\nu=\log N$.
Let now $a(\nu)=\log A(\nu)$ and $b(\nu)=\log B(\nu)$.
Note that all functions are expected to tend to infinity as $\nu=\log N$ tends to
infinity and that $\log(\pi_{K_0}(B(\nu)))=b(\nu)(1+o(1))$ and
$\log(\pi_{K_1}(B(\nu)))=b(\nu)(1+o(1))$ by Chebotarev's density theorem.
We
have 
\begin{align*}
    \log C_{\text{search}} &= O(1)+2a(\nu),\\
    \log C_{\text{linear algebra}}
    &=O(1)+\log\left((\pi_{K_0}(B(\nu))+\pi_{K_1}(B(\nu)))^2\right)\\
    &=O(1)+2b(\nu),\\
    \log (C_{\text{search}}+C_{\text{linear algebra}})
    &=O(1)+\log\max(C_{\text{search}},C_{\text{linear algebra}})\\
    &=O(1)+2\max(b(\nu),a(\nu)).
\end{align*}

Likewise, the constraint on the probabilities can be rewritten, using the notation
$p(u,v)$ from Definition~\ref{def:psi} :
\begin{align*}
    \displaystyle 2 a(\nu)
    + \sum_{i \in \{0,1\}} p(\log M_i, b)
    &= \log(\pi_{K_0}(B(\nu)) + \pi_{K_1}(B(\nu)))\\
    &= \log\left(2(1+o(1))B(\nu)\right)\\
    &= O(1) + b(\nu).
\end{align*}

Given the inaccuracy that exists in Formula~\eqref{eq:nfs-cplx-1}, we
can profit from some
simplifications before we formulate our optimization problem. In light of
this, it is sufficient to minimize the function
$\max(a(\nu),b(\nu))$, as this would imply a cost that would be at most within a constant
factor of the optimal. Likewise, the $O(1)$ in the constraint can be
dropped. Using this latter argument, we can also take into account only
the important parts of the bounds $M_0$ and $M_1$, namely $mA$ and
$mA^d$. Our optimization problem is therefore rewritten as the
following simplified problem, which is our main target here as well as in
Section~\ref{sec:refined_compl}.

\begin{myproblem}[Simplified optimization problem]\label{pb:optim2}
  Find three functions, \[a(\nu), b(\nu), d(\nu) :\bR_{>0}\rightarrow\R_{>0}\] which 
for all $\nu\in\bR_{>0}$ minimize
    $\max(a(\nu), b(\nu))$
subject to the constraint
\begin{equation}\label{eq:main_constraint} p\left(a+ \nu/d, b\right) + p\left(d\,a+ \nu/d,
b\right)+2\,a - b = 0.\end{equation}
\end{myproblem}

To deal with the optimization problem, the usual analysis of NFS relies
on the following result.
\begin{proposition}[Canfield-Erd\H os-Pomerance \cite{CEP} and De
    Bruijn \cite{DeBruijn}]
    \label{prop:cep1}
    Let $\varepsilon \in ]0,1[$. If $3 \leq x/y \leq (1-\varepsilon)x/\log x$ then, as $x/y \rightarrow + \infty$:
    \[p(x, y) = -(x/y)\cdot \log(x/y)\cdot(1+o(1)).\] 
\end{proposition}

Using this low-order result, the analysis is completely classical, and
yields the following well-known expressions, see e.g.~\cite[§11]{BuLePo93}:
\begin{proposition}
\label{prop:first_term}
Let $(a,b,d)$ be a minimizer for Problem~\ref{pb:optim2}.
Then
    \begin{align*}
      a(\nu)&=(8/9)^{1/3}\nu^{1/3}(\log \nu)^{2/3}(1+o(1)),\\
        b(\nu)&=(8/9)^{1/3}\nu^{1/3}(\log \nu)^{2/3}(1+o(1)),\\
        d(\nu) &= (3 \nu/\log \nu)^{1/3}(1+o(1)).\\
        \intertext{Furthermore, }
        (a(\nu)+\nu/d(\nu))/b(\nu)&=\frac12 (3 \nu/\log \nu)^{1/3}(1+o(1)),\\
        (d(\nu)a(\nu)+\nu/d(\nu))/b(\nu)&=\frac32 (3 \nu/\log
        \nu)^{1/3}(1+o(1)). 
    \end{align*}
\end{proposition}

\section{Smoothness}\label{sec:smoothness}

In the following paragraphs we will encounter multiple times two functions,
which we denote by
$\cX:\eta\mapsto \log_2 \eta/\log\eta$ and
$\cY:\eta\mapsto 1/\log\eta$. 
 
The notation $\bR[[X,Y]]$ is used for bivariate formal series with real
coefficients,
and for $\bfS\in\bR[[X,Y]]$ and $i \in \mathbb{Z}_{\geq 0}$, the notation
$\bfS^{(i)}$ stands for the truncation of \bfS\ to total degree less than
or equal to $i$.

We introduce the following class of functions, to capture the asymptotic
behavior of several functions of interest at infinity.

\begin{definition}\label{def:class-C} The class of functions $\cC$ is the set of
functions $f$ defined over a neighborhood of $+\infty$ with values in~$\mathbb R$ such that 
  \[\exists\,
    \bfF \in \bR[[X,Y]],\ \forall n \in \bZ _{\geq 0},
    f(\eta) = \bfF^{(n)}(\cX, \cY) + o(\cY^n).\]
We say that $\bfF$ is the series associated to $f$.
\end{definition}

An important property of $\cC$ that we will use intensively in Section
\ref{sec:refined_compl} is that if $f \in \cC$ and its associated series
\bfF\ satisfies $\bfF(0,0) \neq 0$ then $1/f$ stays in $\cC$. If moreover
$\mathbf F(0,0)$ is positive then $\log f$ is also in $\cC$.

Definition~\ref{def:class-C} deserves several comments. First, the map
from $f\in\cC$ to its associated series \bfF\ is clearly not injective,
since for example the function $1/\eta$ is in \cC, and its associated
series is zero. It is also important to notice that the truncation
$\bfF^{(n)}$ that appears in Definition~\ref{def:class-C} is necessary,
as the \emph{evaluation} of \bfF\ at a given value need not make sense:
the series \bfF\ is only intended to capture the asymptotic behavior of the function $f$ at
infinity.\footnote{In fact (but we will not need this), it is easy to extend
Borel's lemma to show that the link from $\cC$ to \bfF\ is
surjective: Any series $\sum_{i,j}a_{i,j}X^iY^j$, regardless of any notion of convergence, can be
\iftrue
realized as the asymptotic expansion of a function in \cC\ written as
$f(x)=\sum_{i,j}a_{i,j}\tau_{i+j}(x)\cX^i\cY^j$ where
$\tau_n$ is a smooth transition function from $0$ to $1$ whose transition
point is adjusted as a function of $\sum_{i+j=n}|a_{i,j}|$.
\else
realized as the asymptotic expansion of some function in \cC.
To this end,
let $b$ be a positive bump function with support in $[0,1]$, and
$\tau(x)=\int_{-\infty}^x b(t)\mathrm dt/\int_{-\infty}^{+\infty}
b(t)\mathrm dt$. Let $B_n=\sum_{i+j=n}|a_{i,j}|$, $A_n=\exp\exp B_n$,
and 
write
$f(x)=\sum_{i,j}a_{i,j}\tau(x-A_{i+j})\cX^i\cY^j$.
We remark
that
$\tau(x-A_{i+j})\leq(\log\log x)/B_{i+j}$, and the proof that $f$
has the desired asymptotic expansion
follows.
\fi}

\subsection{Smoothness results}

The precision of an asymptotic expansion of the NFS complexity is tightly connected to the precision in results
regarding smoothness probabilities.
Canfield, Erd\H os and Pomerance actually prove a better result than
Proposition~\ref{prop:cep1}.

\begin{theorem}[Canfield, Erd\H os and Pomerance~\cite{CEP}]\label{th:CEP}
For all $x \geq 1$ and for all $u = x/y \geq 3$, we have
        \[p(x,y) \geq -u\log u \cdot p(u),\] where $p(u)=
        1+\mathcal X(u)-\mathcal Y(u)+ \mathcal X(u) \mathcal Y(u)  - \mathcal Y(u)^2 +
    O\left(\mathcal X(u)^2 \mathcal Y(u)\right).$
\end{theorem}

In fact, it turns out that Theorem \ref{th:CEP} is at the same time too strong
and too weak for our purposes. On the one hand, it is too bad that the
asymptotic expansion
of the right-hand side of the inequality stops at some point. On the other hand it is a really
strong result since it is true for basically all $x,u$ without any restriction.
But in the NFS context, we know the magnitudes of $x,u$: We do not need a
smoothness result that would be unconditionally true. A better approximation of
the smoothness probability in a narrower range would suit us.  Hildebrand
proves such a result.

\begin{theorem}[Hildebrand~\cite{Hildebrand}]
\label{th:rho_proba}
Let $\varepsilon>0$. 
For
  $1 \leq u \leq \exp((\log y)^{3/5-\varepsilon})$ and $x=y^u$, we
    have
  \[\frac{\Psi(x,y)}x = \rho(u)\left(1+O\left(\frac{\log(u+1)}{\log y}\right)\right),\]
where $\rho$ is the Dickman--de Bruijn function.
\end{theorem}
Under the Riemann
hypothesis this result even holds in a wider range~\cite[p. 290]{Hildebrand}. 
In the NFS context we are in
the appropriate range to apply Theorem~\ref{th:rho_proba}. Indeed, based
on Proposition~\ref{prop:first_term}, we expect to use
Theorem~\ref{th:rho_proba} in a context where $\log
y = b = \Theta(\nu^{1/3} (\log \nu)^{2/3})$ and $u =
\Theta((\nu/\log\nu)^{1/3})$:
$u$ is polynomial in $\log y$, while the bound in
Theorem~\ref{th:rho_proba} is subexponential.

Asymptotically, the result of Theorem~\ref{th:rho_proba} is therefore
very precise. This leads us to study more precisely the expansion of
$\rho$.

\subsection{Asymptotic expansion of the Dickman--de Bruijn function} In~\cite{rhodev}, De Bruijn proves the following formula :
\[\rho(u) \underset{u \rightarrow + \infty}{\sim} \frac{e^{\gamma}}{\sqrt{2 \pi
u}} \times \mathrm{exp} \left(- \int_0^{\xi} \frac{se^s -e^s+1}{s}
\mathrm ds \right)\] for all $u >1$ and $\xi = (e^u-1)/u$.
Let $\eta=(e^s-1)/s$, so that $s = \log (1+s \eta)$. By substitution we have:
\[\int_0^{\xi} \frac{se^s -e^s+1}{s} \mathrm ds = \int_1^u s \mathrm d \eta.\]

Let us study the integral on the right-hand side. First, we need to know more about $s$.

\medskip

\begin{proposition}
\label{prop:dev_s}
The function $\eta \rightarrow
    s(\eta)/\log \eta$ is in $\cC$.
\end{proposition}

\begin{proof}
We prove by induction on $n$ that there exists $P_n \in \bR[X,Y]$ such
    that, as $\eta \rightarrow +\infty$, we have $s = \log \eta\cdot
    (P_n(\cX, \cY)+ o(\cY^n))$.
    First we show that $s = (\log
    \eta)(1+o(1))$ when $\eta \rightarrow +\infty$, which is to say
    $P_0=1$.
    Indeed $\varphi: s \longmapsto (e^s-1)/s$ is bijective on $\R_{>0}$
    and
    strictly increasing, and so is $\varphi^{-1}$. Since $\varphi(\log
    \eta) = (\eta-1)/\log \eta <  \eta = \varphi(s)$, we have $\log \eta
    <s$. Let now $\varepsilon >0$.
    When $\eta$ is large enough, we have
    $\eta^{1+\varepsilon} > (1+\varepsilon)\eta\log \eta +1$.
    This leads to $\varphi(s) < \varphi((1+\varepsilon)\log \eta)$,
    hence to  $s < (1+\varepsilon)\log \eta$ and proves the base case.

We now proceed with the induction step. Since $s = \log(1+s
    \eta)$,
    we may write
    \begin{align*}
        s
        &= \log s + \log\eta + \log(1+1/(s\eta))\\
        &= \log\eta\cdot\left(1+\log s/\log\eta+o(1/\eta)\right)\\
        &=\log\eta\cdot\left(
            1
            +\cX
            +\cY\cdot
                \log 
                P_n(\cX,\cY)
            +o(\cY^{n+1})\right),
    \end{align*}
    where the last expression omits some terms that are swallowed by $o(
    \cY^{n+1})$.
    Since the constant coefficient of $P_n$ is~$1$,
    an expansion of the inner logarithm above to $n$ terms
    gives the desired result. More precisely, it is easy to verify that
    $P_{n+1}$ is the truncation to
    total degree at most $n+1$ of
    \[1+X-Y\sum_{i=1}^{n}(-1)^i\frac{(P_n-1)^i}{i}.\]
\end{proof}

In fact, a more explicit version of the series associated to $s(\eta)/\log\eta$
can be computed:
\begin{proposition}
\label{prop:s_formula} 
  Letting $\bf P$ denote the series associated to $s(\eta)/\log\eta$, we have
    \[\bfP(X,Y)=
    1+X+Y\cdot
    \sum_{i=1}^{+\infty}\sum_{j=1}^i\frac{S(i,i-j+1)}{j!}X^jY^{i-j}\]
    where the $S(i,j)$ are signed Stirling numbers of the first
    kind.
\end{proposition}

\begin{proof}
By \cite{Comtet70} (see also \cite[Sec.~5,
  Exercise~22]{comtet}), the reciprocal $g(y)$ of $f: x \mapsto e^x/x$
  has the following asymptotic series expansion at the neighborhood of
  $\infty$:
  \[\log y + \log_2 y + \sum\limits_{i=1}^{+ \infty} \sum\limits_{j=1}^i
  \frac{S(i,i-j+1)}{j!} \frac{(\log_2 y)^j}{(\log y)^i}.\]

\smallskip

  By definition, $s(y)$ is the reciprocal of the function $h:x \mapsto e^x/x -1/x$. We
  will show that $s(y)-g(y) \sim (y \log y)^{-1}$ as $y \rightarrow + \infty$. This
  will conclude our proof since it will ensure that for all $n \in \mathbb{Z}_{\geq
  0}$, 
  \[s(y) =  \log y + \log_2 y + \sum\limits_{i=1}^{n} \sum\limits_{j=1}^i
  \frac{S(i,i-j+1)}{j!} \frac{(\log_2 y)^j}{(\log y)^i} + o\left(\frac{1}{(\log
  y)^n}\right)\] as $y$ grows.
    
\smallskip

First, using the
  fact that $h(s(y)) = e^{s(y)}y^{-1}-s(y)^{-1} = y$, we get that $f(s(y)) - f(g(y)) =
  e^{s(y)}s(y)^{-1} - y = s(y)^{-1}$. By the mean value theorem, there exists $\theta_y$
  between $g(y)$ and $s(y)$ such that $s(y)^{-1} = f(s(y)) - f(g(y)) = f'(\theta_y)
  (s(y)-g(y))$. Since $s(y)=\log y+\log_2 y+o(1)$ and $g(y)=\log y+\log_2 y+o(1)$, we deduce that $\theta_y = \log y + \log_2 y + o(1)$ when $y
  \rightarrow +\infty$, hence $f'(\theta_y) \sim y$.
  Consequently, $s(y)-g(y) \sim (ys(y))^{-1}\sim (y\log y)^{-1}$.
\end{proof}

\begin{proposition} \label{prop:radius}
    The series
    \[\sum\limits_{i=1}^{+\infty} \sum\limits_{j=1}^i
    \frac{S(i,i-j+1)}{j!} \cX^j\cY^{i-j}\]
    converges uniformly for $\eta\in[176, +\infty[$.
\end{proposition}

\begin{proof}
    First, for all $k \in \mathbb{Z}_{\geq 0}$ and all $\eta >
    e^e$ we have $\cY\leq\cX$,
    so that
    
  \[\sum_{i=1}^k \sum_{j=1}^i \frac{s(i,i-j+1)}{j!} \cX^j\cY^{i-j} \leq
    \sum_{i=1}^k \sum_{j=1}^i \frac{s(i,i-j+1)}{j!} \cX^i =  \sum_{i=1}^k
    c_i \cX^i\]
where $c_i = \sum\limits_{j=1}^i \frac{s(i,i-j+1)}{j!}$.
Let $a_i = i! \sum\limits_{j=1}^i
    \left|\frac{s(i,i-j+1)}{j!} \right|$. According to the asymptotic formula for
    the sequence \textsf{A138013} of the
    OEIS\footnote{\url{https://oeis.org/A138013}}, we have, as $i \rightarrow +
    \infty$: \[a_i \sim \sqrt{-1-W(-1,-e^{-2})} \times
    (-W(-1,-e^{-2}))^i \times i^{i-1}/e^i\]
    where $W$ is the Lambert $W$ function, so that
    \[\frac{a_i/i!}{a_{i+1}/(i+1)!}\underset{i\rightarrow\infty}{\longrightarrow}
    \frac1{-W(-1,-e^{-2})}.\]
    
    This shows that the power series with coefficients $a_i/i!$ has a
    finite radius of convergence equal to $-1/W(-1,-e^{-2})$. Since
    $|c_i| \leq a_i/i!$, the power series $\sum c_i X^i$ also has finite
    radius of convergence, which is at most as large as the former. Therefore
    the series converges uniformly for
    $\frac{\log\log \eta}{\log \eta}\leq -1/W(-1,-e^{-2})$, and this inequality
    is satisfied for $\eta \geq 176$.

\end{proof}

We now shift gears and study the asymptotic behavior of $\int_1^u s
\mathrm d
\eta$ as $u$ grows. The first step towards this goal lies in the following
proposition.

\begin{proposition}
\label{prop:dev_integral_s}
The function $u \rightarrow\frac{1}{u \log u} \cdot \int_e^u s \mathrm{d}\eta$
  is in $\cC$, and the coefficients of its associated series $\bfQ$ can be
  computed explicitly.

\end{proposition}

\begin{proof}
We prove that for all $n \in \mathbb{Z}_{\geq 0}$, there is a polynomial $Q_n$ such that, as $u \rightarrow + \infty$:
   \[\int_e^u s \mathrm{d}\eta = u \log u \cdot \left(
   Q_n \left(\cX(u), \cY(u)) \right) + o \left( \cY(u)^n \right) \right),\]
   and that for all $i,j\geq 0, Q_i(X, Y) \equiv Q_j(X,Y)\bmod \langle X,
   Y\rangle^{\min(i,j)}$. We emphasize that our proof provides an explicit method to compute $Q_n$.

    Let us first notice that since we are looking for an asymptotic
    expansion of a divergent integral as $u\rightarrow\infty$, up to
    terms that also tend to infinity, we are
    free to choose the lower bound of the integral.

    Let $\Delta$ be the $\bR$-linear
    operator defined on $\bR[X,Y]$ by $\Delta 1=0$, $\Delta
    X=Y\cdot(Y-X)$, $\Delta Y=-Y^2$, and $\Delta(UV)=U\Delta V+(\Delta
    U)V$. One can check that:
    \[\forall T\in\bR[X,Y],\ \frac{\mathrm d}{\mathrm
    d\eta}T(\cX,\cY)=
    \frac1\eta \cdot (\Delta T)(\cX,\cY).\]
    
    Notice that $\Delta\bR[X,Y]\subset Y\bR[X,Y]$.
    
    We use the properties of $\Delta$ to prove an
    intermediate result.
    Let $T$ be an arbitrary bivariate polynomial in $\bR[X,Y]$. Using the above notation, repeated integration by
    parts yields:
    \begin{align*}
        \int_e^u T(\cX,\cY) \mathrm{d} \eta
        & = [\eta T(\cX,\cY) ]_e^u
    - \int_e^u 
        (\Delta T)(\cX,\cY)
    \mathrm{d} \eta
        \\
        &=\sum_{i=0}^{n-1}(-1)^i[\eta \cdot(\Delta^i T)(\cX,\cY)
        ]_e^u+(-1)^{n}\int_e^u (\Delta^{n} T)(\cX,\cY)\mathrm{d}
        \eta\\
        &=[\eta \cdot R(\cX,\cY)]_e^u+\int_e^u o(\cY^{n-1})\mathrm{d} \eta
    \end{align*}
    for $R=\sum_{i=0}^{n-1}(-1)^i\Delta^i T\in\bR[X,Y]$.
    Note that $R$ is a truncation of $(1+\Delta)^{-1}T$.

    This allows us to quickly conclude. Indeed, using
    Proposition~\ref{prop:dev_s}, for all $n \in \mathbb{Z}_{\geq 0}$:
    \begin{align*}
      \int_e^u s \mathrm{d} \eta & =\displaystyle
        [(\eta\log\eta-\eta)\bfP^{(n)}(\cX,\cY)]_e^u
        \ +\\&\displaystyle\int_e^u(1-\log\eta)\Delta (\bfP^{(n)})(\cX,\cY)\mathrm d\eta%
        +\int_e^u o(\cY^{n-1})\mathrm{d} \eta%
        .
    \end{align*} 
    
        Since $(1-\log\eta)\Delta (\bfP^{(n)})(\cX,\cY)$ can be written
        as $T_n(\cX,\cY)$ for some $T_n\in\bR[X,Y]$, the
        previous result shows that there exists $R_n\in\bR[X,Y]$ such
        that
         \[\int_e^u s \mathrm{d} \eta  =
         [\eta\log\eta\cdot(1-\cY)\cdot \bfP^{(n)}(\cX,\cY)]_e^u
        + [\eta\cdot R_n(\cX,\cY)]_e^u+\int_e^u o(\cY^{n-1})\mathrm{d}
        \eta.\]
    The claimed expression, with $Q_n=(1-Y)\bfP^{(n)}+Y\cdot R_n$, follows from the verification that 
    $\int_e^u o(1/ (\log \eta)^{n-1}) \mathrm{d} \eta = o(u/(\log u)^{n-1})$ as $u
    \rightarrow + \infty$, which is unilluminating but easy. Finally, we notice
    that for all $i,j\geq 0, R_i(X, Y) \equiv R_j(X,Y)\bmod \langle X,
       Y\rangle^{\min(i,j)}$, which implies that 
$Q_i(X, Y) \equiv Q_j(X,Y)\bmod \langle X,
   Y\rangle^{\min(i,j)}$.

   The series \bfQ\ has therefore the following expression, which makes
   it easy to compute \bfQ\ from \bfP.
   \[
       \bfQ=(1-Y)\bfP+Y(1+\Delta)^{-1}(1-Y^{-1})\Delta \bfP.
   \]
\end{proof}

\begin{corollary}
\label{coro:dev_rho_only_rho}
We recall that $\bfQ$ is the series introduced in Proposition~\ref{prop:dev_integral_s}. For all $n \in \mathbb{Z}_{\geq 0}$ we have,
    as $u \rightarrow + \infty$:
    \[\rho(u) = \exp\left( -u \log u \left( \bfQ^{(n)}(\cX(u), \cY(u))
    + o \left(    \cY(u)^{n}  \right)      \right)  \right).\]

\end{corollary} 

\begin{proof}
We use the equivalent of $\rho$ introduced at the very beginning of the section.
    Constant offsets are
    absorbed in the error term that comes from the expansion of $\int^u s\mathrm d\eta$, and the result is a direct consequence of
    Proposition~\ref{prop:dev_integral_s}. 
   \begin{align*}
\rho(u) &  \underset{u \rightarrow +
    \infty}{\sim}\displaystyle\frac{e^{\gamma}}{\sqrt{2 \pi u}} \times \mathrm{exp}
    \left(- \int_1^{u} s \mathrm d\eta \right) \\
         &= \displaystyle\mathrm{exp} \left(- \int_e^{u} s \mathrm d\eta  + O(1)\right)\\
        & = \displaystyle\mathrm{exp}\left(-u \log u \cdot
        \left(\mathbf{Q}^{(n)}(\cX(u),\cY(u)) +o \left(\cY(u)^n \right) \right)
        \right). 
   \end{align*}

\end{proof}

\begin{corollary}
\label{coro:dev_rho_as_smoothness_proba}
    In the optimal parameter range of NFS, the
    asymptotic expansion of the
    smoothness probabilities follows from the previous result.
    For any $n\in\Z_{\geq0}$ we
    have, on both sides (rational and
    algebraic), as $N\rightarrow+\infty$:
    \begin{align*}
        \Psi(M_i,B)/B &= \exp\left( -u \log u \left( \bfQ^{(n)}(\cX(u), \cY(u))
    + o \left(    \cY(u)^{n}  \right)      \right)  \right),
    \end{align*}
    where $u$ denotes $\log M_i/\log B$, which is the size ratio on side $i\in\{0,1\}$.
\end{corollary}

\begin{proof}
    As we argued when justifying the use of
    Theorem~\ref{th:rho_proba}, Proposition~\ref{prop:first_term} tells
    us that the optimal parameter range of NFS leads
    to $u=\Theta((\nu/\log\nu)^{1/3})$, with $\nu=\log N$ as in
    Section~\ref{sec:NFS_background}. It follows that $b=\Theta(u\log
    u)$, so that the right-hand side of Theorem~\ref{th:rho_proba} can be
    written as
    \begin{align*}
        \rho(u)\left(1+O\left(\frac{1}{u}\right)\right)&=\rho(u)\exp\left(-u\log
    u\cdot O\left(\frac1{u^2\log u}\right)\right)\\
        &=\exp\left( -u \log u \left( \bfQ^{(n)}(\cX(u), \cY(u))
    + o \left(    \cY(u)^{n}  \right)      \right)  \right).
    \end{align*}
    
    Put otherwise, a function within $O\left(\frac1{\eta^2\log\eta}\right)$ is
    always in \cC, and its associated series is zero.
\end{proof}

\subsection{Computation of \bfQ}

Propositions~\ref{prop:dev_s} and~\ref{prop:s_formula} are completely
explicit. In Proposition~\ref{prop:dev_integral_s}, the expression of
$\bfQ$
is straightforward to compute as well.
For example, the terms of $\bfQ$ up to degree 3, are:
  \[\bfQ ^{(3)}(X,Y) =  1+X-Y+XY-Y^2-\frac{X^2Y}2 + 2XY^2-2Y^3.\]

\section{Asymptotic expansion of $\xi$}
\label{sec:refined_compl}

This section is devoted to the analysis of the asymptotic behavior of the
unknown function $\xi$ involved in the heuristic complexity of NFS. To this
end, we compute the asymptotic
expansion of the unknown functions written $o(1)$ in Prop.~\ref{prop:first_term}.

\subsection{Extension of the class $\cC$}
In order to express conveniently our results on the asymptotic expansion of the function
$\xi$, we introduce a variant of the class $\mathcal C$ defined in
Section~\ref{sec:smoothness}.

\begin{definition}\label{def:class-Cab} Let $\alpha, \beta \in \mathbb{R}$. The class of functions $\cC^{[\alpha,\beta]}$ is the set of functions $f$ with values in $\mathbb R$ such that $\nu\mapsto \nu^{-\alpha}(\log\nu)^{-\beta} f(\nu)$
is in $\mathcal C$.
\end{definition}

\medskip

In particular, $\mathcal C^{[0,0]} = \mathcal C$. If $f$ is in
$\cC^{[\alpha,\beta]}$, then we denote by $\mathbf F\in\mathbb R[X,Y]$ the series
associated to $\nu\mapsto \nu^{-\alpha}(\log\nu)^{-\beta} f(\nu)$ and
call it, by extension, the series associated to~$f$. Note however that
this extension deserves some caution, as the association
makes sense only relative to a given $(\alpha,\beta)$.

The introduction of these classes $\mathcal
C^{[\alpha,\beta]}$ is motivated by the fact that in the context of NFS,
the function $u\mapsto -\log\rho(u)$, which gives the
relevant smoothness probabilities, is in $\cC^{[1,1]}$ (see
Corollaries~\ref{coro:dev_rho_only_rho}
and~\ref{coro:dev_rho_as_smoothness_proba}).
The following proposition establishes stability properties for elements of the classes $\cC^{[\alpha, \beta]}$. 

\begin{proposition} \label{prop:stability}
Let $f\in\mathcal C^{[\alpha,\beta]}$ with associated
series $\mathbf F\in\mathbb R[[X,Y]]$ such that $\alpha> 0$ and $\mathbf F(0,0) > 0$. Then
\begin{enumerate}
  \item $\log f\in\mathcal C^{[0,1]}$ and its associated
    series is
    $\alpha+\beta X/\alpha + (Y\log \mathbf F)/\alpha;$
  \item $\cY(f)=(\log f)^{-1}\in\cC$ and its associated
    series is
    \[Y/(\alpha+\beta X/\alpha + (Y\log \mathbf F)/\alpha);\]
\item if $\alpha > 0$, then $\cX(f)=\log\log f/\log f\in \mathcal C$ and its associated
    series is
    \[(X+Y\log(\alpha+\beta X/\alpha + (Y\log \mathbf F)/\alpha))/(\alpha+\beta
    X/\alpha + (Y\log \mathbf F)/\alpha).\]
\end{enumerate}
\end{proposition}

\begin{proof}
The result follows by direct computations.
\end{proof}

\subsection{Two new proven terms for the NFS complexity}
\label{sec:first_terms_dev}

In this section we develop the functions of interest $a,b,d$, defined
in Section~\ref{sec:NFS_background}. The end result of this section is an asymptotic
expansion of the complexity of NFS with two new terms. We
also introduce several reasonings that will be intensively used and
automatized in Section~\ref{sec:coeffs_expansion}.

\begin{theorem}\label{th:first_terms} The minimizers $a,b,d$ satisfy :
  \begin{align*}
    a &= (8/9)^{1/3}\nu^{1/3}(\log\nu)^{2/3}(1+a_{10} \cX(\nu) + a_{01} \cY(\nu) +
    o(\cY(\nu))),\\
    b &= (8/9)^{1/3}\nu^{1/3}(\log\nu)^{2/3}(1+a_{10} \cX(\nu) + a_{01} \cY(\nu) +
    o(\cY(\nu))),\\
    d &= (3\nu/\log\nu)^{1/3}(1+d_{10}\cX(\nu) + d_{01} \cY(\nu) + o(\cY(\nu))),
  \end{align*}
  where $a_{10} = 4/3, a_{01} = -2 \log 2 + \log 3/6 -2, d_{10} = -2/3$ and
  $d_{01} = \log 2 -5 \log 3/6  +1 $.
\end{theorem}

Before proving Thm.~\ref{th:first_terms}, we show that there exist functions
$a, b, d$ as in Thm.~\ref{th:first_terms} which satisfy
Eq.~\eqref{eq:main_constraint}. They will serve as a baseline for our
minimization problem. One could wonder where the constants
$a_{ij}$ and $d_{ij}$ in the statement of Thm.~\ref{th:first_terms}, and
also in 
the
following lemma, come from. To obtain these constants, we used a computer
algebra system to iteratively expand the
constraint~\eqref{eq:main_constraint} and then we minimized the coefficients of the expansions of $a,b,d$ by hand. In the rest of the section, we will always omit the argument $\nu$ of the functions $\cX$ and $\cY$.

\begin{lemma}\label{lem:existence_first_terms}
  There exist functions $a,b,d$ which satisfy Eq.~\eqref{eq:main_constraint} and such that 
  \begin{align*}
    a &= (8/9)^{1/3}\nu^{1/3}(\log\nu)^{2/3}(1+a_{10} \cX + a_{01} \cY + a_{20}
    \cX^2 + a_{11}\cX\cY + a_{02}\cY^2 + o(\cY^2))\\
    b &= (8/9)^{1/3}\nu^{1/3}(\log\nu)^{2/3}(1+a_{10} \cX + a_{01} \cY + a_{20}
    \cX^2 + a_{11}\cX\cY + a_{02}\cY^2 + o(\cY^2))\\
    d   &= (3\nu/\log\nu)^{1/3}(1+d_{10}\cX + d_{01} \cY + o(\cY)),
\end{align*}
  where
  \[\begin{array}{c}a_{10} = 4/3,\quad a_{01} = -2 \log 2 + \log 3/6
      -2,\\a_{20} = -4/9,
      \quad a_{11}
=4\log 2/3-\log 3/9+4,\\
a_{02}=-(\log 2)^2 +(\log 2\cdot\log 3)/6-7(\log 3)^2/36-6\log 2+\log 3/2-5,\\
d_{10} =
  -2/3,\text{\quad and\quad}
  d_{01} = \log 2 -5 \log 3/6  +1.
\end{array}\]
\end{lemma}

\begin{proof}
Set
\begin{align*}
  a &= (8/9)^{1/3}\nu^{1/3}(\log\nu)^{2/3}(1+a_{10} \cX + a_{01} \cY + a_{20}
  \cX^2 + a_{11}\cX\cY + a_{02}\cY^2 + \widetilde{a}\cX^3)\\
  b &= (8/9)^{1/3}\nu^{1/3}(\log\nu)^{2/3}(1+a_{10} \cX + a_{01} \cY + a_{20}
  \cX^2 + a_{11}\cX\cY + a_{02}\cY^2)\\
  d &= (3\nu/\log\nu)^{1/3}(1+d_{10}\cX + d_{01} \cY),
\end{align*} 
where $a_{10}, a_{01}, a_{20}, a_{11}, a_{02}, d_{10}, d_{01}$ are as in the lemma, and $\widetilde{a}$ is an unknown
function of the variable $\nu$.

Given these expressions, we wish to rewrite Eq.~\eqref{eq:main_constraint}
as a function of $\widetilde{a}$. This is particularly tedious, but
straightforward. The only needed tools are formulas in
Prop.~\ref{prop:stability} over the function field $\mathbb
R(\widetilde{a})$. The code repository mentioned in the introduction
of this article shows how the expansion can be carried out with a
    computer algebra system (and the same holds for other calculations in
    this section). We obtain
that Eq.~\eqref{eq:main_constraint} with the functions set above can be
  rephrased as:
\[(\widetilde{a}-32/81) = \varepsilon(\widetilde{a}, \nu),\]
for a continuous function $\varepsilon$ such that for all $t
\in\mathbb R$,
$\lim_{\nu\rightarrow\infty}\varepsilon(t, \nu) = 0$. 
  Let now $t_-=31/81$ and $t_+=33/81$. Since  $\varepsilon(t_-,
  \nu)$ and $\varepsilon(t_+,
  \nu)$ both tend to zero, we can define $\nu_-$ and $\nu_+$ such that
  \begin{align*}
      \forall\nu>\nu_-,&\ (t_--32/81)-\varepsilon(t_-,
  \nu)<0,\\
      \forall\nu>\nu_+,&\ (t_+-32/81)-\varepsilon(t_+,
  \nu)>0.
  \end{align*}
  
  By the
  intermediate value theorem, we obtain that for any
  $\nu>\max(\nu_-,\nu_+)$, there exists 
$t\in [t_-,t_+]$ such that $(t-32/81) =
\varepsilon(t, \nu).$ Let now $\widetilde{a}$ be the function of
  $\nu$ that
returns such a number $t$. Then the functions $a$, $b$, and $d$ defined
  above
 satisfy by construction the desired property.
\end{proof}

\begin{proof}[Proof of Theorem~\ref{th:first_terms}]
The roadmap of the proof is the following:
\begin{enumerate}
  \item We express the constraint
\eqref{eq:main_constraint} using a sufficiently precise asymptotic
        expansion of
	the smoothness probability
(given by Corollary \ref{coro:dev_rho_as_smoothness_proba}) and the
        asymptotic expansions of $a,b,d$ known
so far. Then we prove that the $o(1)$ involved in the asymptotic
        expansions of $a,b,d$
are actually in the class $O(\cX^\lambda\cY^\mu)$ for some $\lambda, \mu
\geq 0$ (not both being zero) so that we can write these $o(1)$ as
$C\cdot\cX^\lambda\cY^\mu$ where $C$ is a function bounded at a neighbourhood of $\infty$. 
  \item We prove that $C = c + o(1)$ where $c$ is a constant computed along the way and restart the whole process using the more precise asymptotic
        expansions
for $a,b,d$ that we have just obtained in order to compute the next
    coefficients.
\end{enumerate}

\smallskip

\textit{Step 1}: By Proposition \ref{prop:first_term}, minimizers can be
written as
\begin{align*}
a &= (8/9)^{1/3}\nu^{1/3}(\log\nu)^{2/3}(1+\widetilde{a}),\\
b &= (8/9)^{1/3}\nu^{1/3}(\log\nu)^{2/3}(1+\widetilde{b}),\\
d &= (3\nu/\log\nu)^{1/3}(1+\widetilde{d}),
\end{align*}
where $\widetilde{a}, \widetilde{b}, \widetilde{d}= o(1)$.

Direct computations and simplifications (that involve Taylor series expansions,
Prop.~\ref{prop:stability}, and the asymptotic expansion of the
  smoothness probability in
Corollary~\ref{coro:dev_rho_as_smoothness_proba}) that take into account the fact that
  $\widetilde{
  a}, \widetilde{b}, \widetilde{d}= o(1)$ rephrase Eq.~\eqref{eq:main_constraint} as 
\[\widetilde{a} = \frac{2}{3} \widetilde{b}{\mathstrut}^{\>2}
+ \frac{1}{3} \widetilde{d}{\mathstrut}^{\>2}   + O(\cX). \]

    The last equation shows that $\widetilde a
    \cX^{-1} $ is bounded below by a finite constant. Moreover,
Lemma~\ref{lem:existence_first_terms} ensures the existence of functions
    $a_0, b_0, d_0$ that can be used as substitutes for $\widetilde a,
    \widetilde b, \widetilde d$ above, that satisfy the
    constraint~\eqref{eq:main_constraint}, and such that $\lim
    {a_0}\cX^{-1}$ exists and is finite.
    Since $a,b,d$ are minimizers of
    Problem~\ref{pb:optim2}, 
    $\widetilde a
    \cX^{-1} $ is also
    upper bounded by ${a_0}\cX^{-1}$.
    Therefore, $\widetilde a\in
    O(\cX)$, whence the same also holds for $\widetilde
    {b}{\mathstrut}^{\>2}$ and $\widetilde
    {d}{\mathstrut}^{\>2}$.

Replacing $\widetilde{a}, \widetilde{b}, \widetilde{d}$ respectively by $\overline
a\cX$, $\overline b\cX^{\frac{1}{2}}$, $\overline
d\cX^{\frac{1}{2}}$ for some functions $\overline a, \overline
b,\overline d$ bounded at a neighborhood of $+\infty$ in Eq.~\eqref{eq:main_constraint}, we obtain that
\[-\overline{a} + \frac{2}{3}
\overline{b}^2+\frac{1}{3}\overline{d}^2+\frac{4}{3}  = o(1),\]
which means in particular that $\overline{a} \geq 4/3+o(1)$.
By minimality of
$a$ we must also have $\overline{a} \leq 4/3+o(1)$ so as not to
contradict Lemma~\ref{lem:existence_first_terms}. So $\overline{a} = 4/3+o(1)$ and
we must necessarily have $\overline{b} = \overline{d} =o(1)$. Therefore, we obtain
\begin{align*}
a &= (8/9)^{1/3}\nu^{1/3}(\log\nu)^{2/3}(1+ 4\cX/3 +\widetilde
a\cX),\\
b &= (8/9)^{1/3}\nu^{1/3}(\log\nu)^{2/3}(1+\widetilde{b}\cX^{\frac{1}{2}}),\\
d &= (3\nu/\log\nu)^{1/3}(1+\widetilde{d}\cX^{\frac{1}{2}}),
\end{align*} for some fresh functions $\widetilde{a}, \widetilde{b}, \widetilde{d}= o(1)$. 

\medskip

\textit{Step 2}: We use the result obtained in Step~1 and the asymptotic
expansion of the smoothness probability
(Corollary~\ref{coro:dev_rho_as_smoothness_proba})
to deduce by direct computations the following equality enforced by Eq.~\eqref{eq:main_constraint}:
\[\left(-\widetilde{a} + \frac{2}{3}
\widetilde{b}{\mathstrut}^{\>2} +
\frac{1}{3}\widetilde{d}{\mathstrut}^{\>2} \right) \cdot\cX = O(\cY).\] 

Following the
same reasoning as in Step~1, $\widetilde{a} \cX\cY^{-1}$, $ \widetilde{b}
\cX^{\frac{1}{2}} \cY^{-\frac{1}{2}}$ and $\widetilde{d}\cX^{\frac{1}{2}}
\cY^{-\frac{1}{2}}$ must be bounded at a neighborhood of $+\infty$. Replacing
$\widetilde{a}$ (resp.~$\widetilde{b}$, $\widetilde{d}$) by $\overline
a\cX^{-1}\cY$ (resp.~$\overline
  b\cX^{-\frac{1}{2}}\cY^{\frac{1}{2}}$ and $\overline
  d\cX^{-\frac{1}{2}}\cY^{\frac{1}{2}}$) for some functions $\overline
  a,\overline b,\overline d$ bounded at a neighborhood of $+\infty$,
  we obtain the following equality:   
\[-\overline a + \frac{2\overline b^2}3 + \frac{\overline d^2}3 - 2\log 2 +
\frac{\log 3}6 - 2
= o(1).\] 

By minimality and using our baseline result
Lemma~\ref{lem:existence_first_terms} as we did in Step 1, the equalities
$\overline{a} = -2\log 2 +  \log 3/6 -2 + o(1)$,
$\overline{b} = o(1)$,
and $\overline{d} = o(1)$
must hold, which means that:
\begin{align*}
a &= (8/9)^{1/3}\nu^{1/3}(\log\nu)^{2/3}(1+ 4\cX/3+ (-2\log 2 +
\log 3/6 -2)\cY +\widetilde{a}\cY),\\
b &= (8/9)^{1/3}\nu^{1/3}(\log\nu)^{2/3}(1+\widetilde{b}\cY^{\frac{1}{2}}),\\
d &= (3\nu/\log\nu)^{1/3}(1+\widetilde{d}\cY^{\frac{1}{2}}),
\end{align*} for some fresh functions $\widetilde{a}, \widetilde{b}, \widetilde{d}= o(1)$. 

\medskip

\textit{Step 3}: Again we use the asymptotic expansion obtained in Step~2 to refine
the asymptotic equality from Eq.~\eqref{eq:main_constraint} and
obtain:
\begin{align*}
  \left(\left(-2\log(2) + \frac{5\log(3)}{3} - 2+o(1)\right)\widetilde{d}+ \left(8 \log(2) - \frac{2\log(3)}{3} +
  2+o(1)\right)\widetilde{b}\right)\frac{\cY^{\frac{3}{2}}}{3} &\\
  +\left(-\widetilde{a}+ \frac{\widetilde{d}{\mathstrut}^{\>2}}{3} + \frac{2 {\widetilde
  b}{\mathstrut}^{\>2}}{3}\right)\cY +
  \left(\frac{4\widetilde{d}}{9} - \frac{16\widetilde{b}}{9}\right)\cX\cY^\frac{1}{2}= O(\cX^2).
\end{align*}
where $o(1)$ are explicit expressions in $\widetilde{a}, \widetilde{b},
\widetilde{d}$, which we omit for brevity.

This can be rephrased as 
\begin{align*}
\dfrac{1}{3}\underbrace{\left(\widetilde{d}\cY^{\frac{1}{2}}+\frac{2}{3}\cX+
  \left(-\log(2) + \frac{5\log(3)}6 - 1+o(1)\right)\cY\right)^2}_{ = \delta(\nu)} +&\\
\dfrac{2}{3}\underbrace{\left(\widetilde{b}\cY^{\frac{1}{2}}-\frac{4}{3}\cX+ 
  \left(2 \log(2) - \frac{\log(3)}6 + \frac{1}{2}+o(1)\right)\cY\right)^2}_{ = \beta(\nu)}
-\widetilde{a} \cY &= O(\cX^2).
\end{align*}

We prove now that $\beta(\nu)=O(\cX^2)$ and $\delta(\nu)=O(\cX^2)$. Assume by
contradiction that this does not hold. Then $\widetilde{a}(\nu)$ is
positive asymptotically and it cannot belong to the class $O(\cX^2)$. This
would contradict our upper bound for the minimum given in
Lemma~\ref{lem:existence_first_terms}.

Consequently, $\beta(\nu)$ and $\delta(\nu)$ belong to $O(\cX^2)$ and then so
does $\widetilde{a} \cY$. This means that 
$\widetilde{a} = O(\cX^2\cY^{-1})$ and $\widetilde{b}, \widetilde{d} =
O(\cX\cY^{-\frac{1}{2}})$. As usual we call $\overline{a}, \overline{b},
\overline{d}$ the functions $\widetilde{a}\cX^{-2}\cY,
\widetilde{b}\cX^{-1} \cY^{\frac{1}{2}}$ and
$\widetilde{d}\cX^{-1}\cY^{\frac{1}{2}}$ bounded at $\infty$, and we substitute in Eq.~\eqref{eq:main_constraint} to get

\begin{equation}\label{eq:constraint_coeffs}\left(-\overline a-\frac49\right)+
  \frac23\left(\overline b-\frac43\right)^2+
  \frac13\left(\overline d+\frac23\right)^2=o(1).\end{equation}

Lemma~\ref{lem:existence_first_terms} ensures that we must have $\overline{a}
\leq -4/9+o(1)$, otherwise it would contradict the minimality of $a$. The
equation above ensures that we must have $\overline{a} \geq -4/9+o(1)$ as well.
So $\overline{a} = -4/9+o(1)$, which implies that $\overline{b} = 4/3+o(1)$ and
$\overline{d} = -2/3+o(1)$. This third step of this proof gives:
\begin{align*}
a &= (8/9)^{1/3}\nu^{1/3}(\log\nu)^{2/3}(1+ 4 \cX/3+ (-2\log 2 +
\log 3/6 -2)\cY-4\cX^2/9 +\widetilde{a}\cX^2),\\
b &= (8/9)^{1/3}\nu^{1/3}(\log\nu)^{2/3}(1+4\cX/3+\widetilde{b}\cX),\\
d &= (3\nu/\log\nu)^{1/3}(1-2\cX/3+\widetilde{d}\cX),
\end{align*} for some fresh unknown functions $\widetilde{a}, \widetilde{b}, \widetilde{d}= o(1)$.

\medskip

\textit{Step 4}: Substituting the asymptotic expansion obtained in Step~3
yields 
\[\left(-\widetilde{a} +
\frac{2\widetilde{b}{\mathstrut}^{\>2}}3 +
\frac{\widetilde{d}{\mathstrut}^{\>2}}3 \right)\cdot\cX^2 = O(\cX\cY).\] 

As in Step~1, $\widetilde{a}$ must belong to $O(\cX^{-1}\cY)$ and
$\widetilde{b}, \widetilde{d}$ to $O(\cX^{-\frac{1}{2}}\cY^{\frac{1}{2}})$ so
as not to contradict Lemma~\ref{lem:existence_first_terms}. Using the notations
$\overline{a}, \overline{b}, \overline{d}$ for the asymptotically bounded
functions $\widetilde{a}\cX\cY^{-1},
\widetilde{b}\cX^{\frac{1}{2}}\cY^{-\frac{1}{2}}$ and
$\widetilde{d}\cX^{\frac{1}{2}}\cY^{-\frac{1}{2}}$, we find by substitution in
Eq.~\eqref{eq:main_constraint} that $\overline{b} = o(1)$, $ \overline{d}=o(1)$ and
$\overline{a} = 4 \log 2/3  - \log 3/9 +4 +o(1)$, \emph{i.e.},
\begin{align*}
a =& (8/9)^{1/3}\nu^{1/3}(\log\nu)^{2/3}(1+ 4\cX/3+ (-2\log 2 + \log 3/6 -2)\cY \\
& -4\cX^2/9+(4 \log 2/3 - \log 3/9
+4)\cX\cY+\widetilde{a}\cX\cY),\\
b =& (8/9)^{1/3}\nu^{1/3}(\log\nu)^{2/3}(1+4
\cX/3+\widetilde{b}\cX^{\frac{1}{2}} \cY^{\frac{1}{2}}),\\
d   =&
(3\nu/\log\nu)^{1/3}(1-2\cX/3+\widetilde{d}\cX^{\frac{1}{2}} \cY^{\frac{1}{2}}),
\end{align*}
for some fresh unknown functions $\widetilde{a}, \widetilde{b}, \widetilde{d}= o(1)$.

\medskip

\textit{Step 5}: We expand the constraint for the last time and after a
factorization
that follows the pattern of Step 3 we get:
\begin{align*}
  \frac13\left(\widetilde{d}\cX^{\frac{1}{2}}+\left(-\log 2 + \frac{5\log 3}6 -
  1\right)\cY^{\frac{1}{2}}\right)^2\cY + &\\
  \frac23\left(\widetilde{b}\cX^{\frac{1}{2}} +
  \left(2\log 2 - \frac{\log 3}6 + \frac12\right)\cY^{\frac{1}{2}}\right)^2\cY -
  \widetilde{a}\cX\cY &= O(\cY^2).
\end{align*}

Again, we must have $\widetilde{a} = O(\cX^{-1} \cY)$ and $\widetilde{b}, \widetilde{d} =
O(\cX^{-\frac{1}{2}}\cY^{\frac{1}{2}})$, so as not to contradict
Lemma~\ref{lem:existence_first_terms}. We also compute the associated limits
for $\widetilde b$ and $\widetilde d$
by using the same method as in the previous steps: we let $\overline a$,
$\overline b$, $\overline d$ denote the bounded functions $\widetilde
a\cX\cY^{-1}$, $\widetilde{b}\cX^{\frac{1}{2}}\cY^{-\frac{1}{2}}$, 
$\widetilde{d}\cX^{\frac{1}{2}}\cY^{-\frac{1}{2}}$.
Direct computations yield
\begin{align*}\left[\frac 23\left(\overline b+2\log 2-\frac{\log 3}6+\frac
  12\right)^2 -\frac 32\right]+ \frac 13\left(\overline d+\frac{5\log
3}6-\log 2-1\right)^2=\\\left(\overline a-\left(-(\log 2)^2 + \frac{\log 2\log 3}6 -
\frac{7 (\log 3)^2}36 - 6\log 2 + \frac{\log 3}2 - 5\right)\right) +
o(1).\end{align*}

Lemma~\ref{lem:existence_first_terms} ensures that $\overline
b\leq a_{01}+o(1)$ and $\overline a\leq a_{02}+o(1)$.
This implies that the lefthand side of the equality is bounded below by some
function in $o(1)$ and hence 
$\lim\overline a = a_{02}$. Therefore we get that
\begin{align*}
\lim \widetilde{b}\cX^{\frac{1}{2}}\cY^{-\frac{1}{2}} &= -2\log 2 + \log 3/6 - 2,\\
\lim \widetilde{d}\cX^{\frac{1}{2}}\cY^{-\frac{1}{2}} &= \log 2 - 5\log 3/6 + 1,
\end{align*}
which concludes the proof.
\end{proof}

\begin{corollary} \label{coro:NFS_two_terms} Let
    $\Phi(\nu)=\max(a(\nu),b(\nu))$ be the quantity minimized by
Problem~\ref{pb:optim2}. The heuristic complexity
$C(N)=\exp(2\Phi(\log N))$ to factor an integer $N$ with NFS satisfies:  
\[\log C(N) = \sqrt[3]{\frac{64}{9}}(\log N)^{1/3} (\logTwo N)^{1/3}\left(1+a_{10} \frac{\logThree
N}{\logTwo N} +  \frac{a_{01}}{\logTwo N} +o\left(\frac{1}{\logTwo
N}\right)\right)\] where $a_{10} = 4/3$ and $a_{01} = -2\log 2 +\log 3/6-2$.
\end{corollary}

A natural question to ask is whether the process used in the proof of Thm.~\ref{th:first_terms} can be continued. As we will see in the next section, the answer to this question is yes. 

\subsection{Further terms in the asymptotic expansion of the complexity of NFS}
\label{sec:coeffs_expansion}

In this section, we describe how the arguments of
Section~\ref{sec:first_terms_dev} can be turned into three algorithms that
allow to compute more precise asymptotic expansions for the minimizers $a,b,d$. These
three algorithms take as input the precision requested for the asymptotic
expansion of the complexity and they mirror the three steps used in the proof of
Theorem~\ref{th:first_terms}: 
\begin{itemize}
\item Algorithm \textsc{GuessTerms} guesses the asymptotic expansion of the
  minimizers of Problem~\ref{pb:optim2}.
\item Algorithm \textsc{ProveExistence} proves the existence of functions
  satisfying the constraint in Problem~\ref{pb:optim2}, and whose asymptotic
  expansion is the output of Algorithm~\textsc{GuessTerms}. This
  establishes an upper bound for the minimum of the objective function in Problem~\ref{pb:optim2}.
\item Algorithm \textsc{ProveMinimality} proves that the asymptotic
  expansion of the minimizers of Problem~\ref{pb:optim2} must be the
  output of Algorithm \textsc{GuessTerms}.
\end{itemize}

These algorithms are used in the following way. We set a degree $n>1$. The three
algorithms are used to prove that $\xi(N) = Q(\cX(\log N), \cY(\log N)) +
o(\cY(\log N)^n)$
where $Q$ is a bivariate polynomial of total degree $n$, whose
coefficients
are computed along the way. We emphasize that these algorithms might
fail, \emph{i.e.}, become unable to compute or prove new terms at some point.
We were unable to prove that our algorithms never fail, but experimentally they
never did. And even if one of them does, the terms of the complexity computed
up until the failure point are guaranteed to be correct. We also point out that
all the bivariate polynomials that we consider in these algorithms have
coefficients in $\mathbb Q(\log 2, \log 3)$, so they can be described exactly
without having to rely on floating-point computations.

\medskip

We now describe
more precisely the algorithms \textsc{GuessTerms},
\textsc{ProveExistence} and \textsc{ProveMinimality} assuming there is no
failure.

\medskip

Algorithm \textsc{GuessTerms} assumes that $a,b,d$ belong respectively to the
classes $\cC^{[1/3,2/3]}, \cC^{[1/3,2/3]}, \cC^{[1/3,-1/3]}$ and even that
$a=b$. We call $\mathbf A, \mathbf B, \mathbf D$ the bivariate series
associated to $a,b,d$. First, we expand
Eq.~\eqref{eq:main_constraint} based on the asymptotic
expansions of $a,b$ and $d$ (initialized thanks to the result of
Theorem~\ref{th:first_terms}). Then we minimize the leading term of the
asymptotic expansion of the objective function in Problem~\ref{pb:optim2} under
the constraint that the coefficient of the main term in the asymptotic
expansion of Eq.~\eqref{eq:main_constraint} must vanish. This is done thanks to
Taylor series expansions at infinity and bivariate series computations at finite precision. From this, we deduce the
next coefficients of our series. The repetition of this process provides an algorithm
to iteratively compute the coefficients of the series $\mathbf A, \mathbf B,
\mathbf D$. These series will be our candidates for the asymptotic
expansion
of the minimizers of Problem~\ref{pb:optim2}.

\medskip

Algorithm \textsc{ProveExistence} is the counterpart of Lemma~\ref{lem:existence_first_terms}. It checks that there exist functions $\alpha,\beta,\delta$ such that:
\begin{itemize}
  \item The asymptotic expansion of $\alpha$ is $\alpha = \mathbf{A}^{(n+1)}(\cX, \cY) +
    o(\cY^{n+1})$;
  \item The asymptotic expansion of $\beta$ is $\beta = \mathbf{B}^{(n+1)}(\cX,\cY) +
    o(\cY^{n+1})$;
  \item The asymptotic expansion of $\delta$ is $\delta = \mathbf{D}^{(n+1)/2}(\cX,\cY) +
    o(\cY^{(n+1)/2})$;
  \item The functions $\alpha,\beta,\delta$ satisfy the constraint~\eqref{eq:main_constraint}.
\end{itemize}

Algorithm \textsc{ProveExistence} serves the same purpose as
Lemma~\ref{lem:existence_first_terms}: Give a baseline result that ensures
the existence of functions that satisfy the constraint~\eqref{eq:main_constraint}
and whose asymptotic expansions are known up to a given degree. In particular, a
solution to Problem~\ref{pb:optim2} must be smaller than $\alpha, \beta,
\delta$. The algorithm works similarly to the proof of
Lemma~\ref{lem:existence_first_terms} by setting three functions $a,b,d$ where
$a$ depends on an unknown function $\widetilde a$:
\begin{align*}
  a &= (8/9)^{1/3}\nu^{1/3}(\log\nu)^{2/3}(\mathbf A^{(n+1)}(\cX, \cY) +
  \widetilde{a}\cX^{n+2});\\
b &= (8/9)^{1/3}\nu^{1/3}(\log\nu)^{2/3}\mathbf B^{(n+1)}(\cX,\cY);\\
d &= (3\nu/\log\nu)^{1/3}\mathbf D^{(n+1)/2}(\cX, \cY).
\end{align*} 

Then the algorithm checks by using Taylor series expansions that the constraint of
Problem~\ref{pb:optim2} instantiated with these functions
can be rewritten as $(\widetilde a-\kappa)=\varepsilon(\widetilde a, \nu)$ for
some $\kappa\in\mathbb R$ and $\varepsilon$ a function as in the proof of
Lemma~\ref{lem:existence_first_terms}.

\medskip

Algorithm \textsc{ProveMinimality} is the counterpart of the proof of
Theorem~\ref{th:first_terms} and it follows its roadmap. We proceed iteratively
with the terms of the asymptotic expansion, computing and proving the
expansions of the minimizers $a,b,d$ up to degree~$n$. To do so, we use the
current (proven) knowledge of the asymptotic expansions of $a,b,d$, initialized
thanks to Theorem~\ref{th:first_terms}, and expand the constraint by computing
Taylor series expansions at infinity.

To prove that the next terms of the series guessed by Algorithm~\textsc{GuessTerms} are correct, we proceed as follows.
First, 
we
prove that the remainder of the series expansion, which has form
$o(\cX^{\frac i2}\cY^{\frac j2})$, is in fact of the form
$O(\cX^{\frac{i'}2}\cY^{\frac{j'}2})$, where $(i', j')$ is strictly larger than
$(i,j)$ for the graded lexicographical ordering, so that
$O(\cX^{\frac{i'}2}\cY^{\frac{j'}2})$ is a proper subset of $o(\cX^{\frac
i2}\cY^{\frac j2})$. This first step works as long as the equation derived from the constraint follows certain patterns, which are given in
Proposition~\ref{prop:patterns} below.

In a second step, we prove that this remainder has in fact the form
$\kappa\cX^{\frac{i'}2}\cY^{\frac{j'}2}(1+o(1))$, where $\kappa$ is the
corresponding coefficient in the guessed series. This verification involves an equation that is similar to Eq.~\eqref{eq:constraint_coeffs} in the proof of
Theorem~\ref{th:first_terms}. Again, the shape of the equation encountered is crucial to proceed.

Based on our experiments, we surmise that only three patterns occur in the first step, and that the equation encountered in the second step always matches the shape of Eq.~\eqref{eq:constraint_coeffs}. Algorithm  \textsc{ProveMinimality} only consists in verifying that this holds.

\begin{proposition}
  \label{prop:patterns}
  Let $\widetilde a, \widetilde b, \widetilde d = o(1)$ be three functions.
  \begin{itemize}
    \item \textbf{Pattern (P1)}: Let $i> 0, j\geq 0$. Assume that
      $\widetilde a(\nu) \leq a_0(\nu)$, $\widetilde b(\nu)\leq b_0(\nu)$, where $a_0, b_0 =
      O(\cX^{-1}\cY)$. Assume further that 
      \[\left( \frac{{\widetilde d}{\mathstrut}^{\>2}}{3} - \widetilde a - 2\widetilde b \right)\cX^i\cY^j =
      O(\cX^{i-1}\cY^{j+1}).\]
      Then $\widetilde a, \widetilde b=O(\cX^{-1}\cY)$ and $\widetilde
      d=O(\cX^{-\frac 12}\cY^{\frac 12})$.
    \item \textbf{Pattern (P2)}: Let $i>1$ and $\kappa\in\mathbb R$. Assume that
      $\widetilde a(\nu)\leq a_0(\nu)$, $\widetilde b(\nu)\leq b_0(\nu)$, where $a_0, b_0 =
      O(\cX^{i+1}\cY^{-i})$. Assume further that 
      \[\left(\frac{{\widetilde d}{\mathstrut}^{\>2}}{3} - \widetilde a - 2\widetilde b(1+o(1))\right)\cY^i +
      \kappa\widetilde d \cX^{\frac{i+1}2}\cY^{\frac i2}(1+o(1)) =
      O(\cX^{i+1}).\]
      Then $\widetilde a = O(\cX^{i+1}\cY^{-i}), \widetilde b \in
          O(\cX^{i+1}\cY^{-i})\subset
          O(\cY^{\frac 12})$, and $\widetilde d = O(\cX^{\frac{i+1}2}\cY^{-\frac i2})$.
      \item \textbf{Pattern (P3)}: Let $i>2$. Assume that
      $\widetilde a(\nu) \leq a_0(\nu)$, $\widetilde b(\nu)\leq b_0(\nu)$, where $a_0 = O(\cX^{-1}\cY)$
      and $b_0 =O(\cX^{i-1}\cY^{-i+\frac{3}{2}})$. Assume further that 
      \[\left(\frac{{\widetilde d}{\mathstrut}^{\>2}}{3} - \widetilde a \right)\cX^i -2\widetilde
      b\cY^{i-\frac{1}{2}} = O( \cX^{i-1}\cY ).\]
      Then $\widetilde a = O(\cX^{-1}\cY)$ and $\widetilde b
      =O(\cX^{i-1}\cY^{-i+\frac{3}{2}})$, and $\widetilde d =
      O(\cX^{\frac 12}\cY^{-\frac 12})$.
  \end{itemize}
\end{proposition}

\begin{proof} 
    The proof in all three cases is very similar. We only prove 
  Pattern \textbf{(P1)} as an example.

  Dividing the equality by $\cX^i\cY^j$, we get
  \[\widetilde a+2\widetilde b =\frac{{\widetilde d}{\mathstrut}^{\>2}}3+O(\cX^{-1}\cY),\]
  which shows that $\widetilde a+2\widetilde b$ is bounded below
  by a function in $O(\cX^{-1}\cY)$.
    Since $\widetilde a + 2\widetilde b \leq a_0+2b_0 =
  O(\cX^{-1}\cY)$, we deduce that $\widetilde a+2\widetilde b =
  O(\cX^{-1}\cY)$. Writing $\overline a, \overline b$ for $\widetilde
  a\cX\cY^{-1}, \widetilde b\cX\cY^{-1}$ respectively, we get that $\overline
  a+2\overline b = O(1)$. Together with the fact that $\overline a,
  \overline b$ are bounded above by $\limsup a_0\cX\cY^{-1}, \limsup
  b_0\cX\cY^{-1}$ respectively, this implies that $\overline a, \overline b\in
  O(1)$, and hence that $\widetilde a,\widetilde b \in O(\cX^{-1}\cY)$.
  Hence ${\widetilde d}{\mathstrut}^{\>2}$ must also belong to
  $O(\cX^{-1}\cY)$.
 
\end{proof}

In fact these three patterns do not appear at random during the proof: They appear according to the shape of the remainder $o(\cX^i\cY^j)$ in the series expansion of $a$ we are currently considering. When $i \neq 0$ and $j \neq 0$ the equation follows pattern \textbf{(P1)}, when $i = 0$ it follows pattern \textbf{(P2)} and when $j = 0$ it follow pattern \textbf{(P3)}. In particular, pattern \textbf{(P3)} is always encountered one step after pattern \textbf{(P2)} has been encountered. We want to consider that $\widetilde{b}$ is in $O(\cY^{\frac{1}{2}})$ instead of the tighter class $O(\cX^{i+1} \cY^{-i})$ at the end of pattern \textbf{(P2)} precisely to ensure that the next step will yield an equation that follows pattern \textbf{(P3)}.

The three algorithms \textsc{GuessTerms}, \textsc{ProveExistence} and
\textsc{ProveMinimality} are used consecutively in order to compute a
proven
asymptotic expansion of $a,b,d$ at a given precision $n$. 
We have implemented a function~\textsc{ComputeProvenExpansion} which performs
this process: It takes as input an integer $n>1$ and --- if it does not fail
--- it returns two bivariate
polynomials $\mathbf A^{(n+1)}, \mathbf D^{(n+1)/2}$ of respective degrees $n+1$ and $(n+1)/2$
such that the minimizers $a,b,d$ of Problem~\ref{pb:optim2} satisfy
\begin{align*}
  a &= (8/9)^{1/3}\nu^{1/3}(\log\nu)^{2/3}(\mathbf A^{(n+1)}(\cX, \cY)+ o(\cY^{n+1})),\\
  b &= (8/9)^{1/3}\nu^{1/3}(\log\nu)^{2/3}(\mathbf A^{(n)}(\cX, \cY)+ o(\cY^{n})),\\
  d &= (3\nu/\log\nu)^{1/3}(\mathbf D^{(n+1)/2}(\cX, \cY)+ o(\cY^{(n+1)/2})).
\end{align*}

The expansion for $b$ is proven up to a degree one less than the one for $a$. The underlying reasons are the same as in the proof of Theorem~\ref{th:first_terms}, where this behavior was first encountered.

This implies that the heuristic complexity of NFS is bounded above by
\[\exp\left[\sqrt[3]{\frac{64}9}(\log N)^{1/3}(\log_2 N)^{
2/3}\left(\mathbf A^{(n)}\left(\frac{\log_3 N}{\log_2 N}, \frac{1}{\log_2
N}\right)+ o\left(\frac
1{(\log_2 N)^n}\right)\right)\right].\]

\section{Experimental results}\label{sec:expe}
In this section, we report on experimental results obtained with our
implementation of the algorithms described in
Section~\ref{sec:coeffs_expansion}. Our implementation is available in the code
repository mentioned in the introduction of this article.

These algorithms provide an asymptotic expansion of the heuristic complexity $C(N)$ of NFS. More precisely, they output
coefficients of a bivariate series $\mathbf A\in \mathbb Q(\log 2, \log 3)[[X,
Y]]$ such that for all $n\geq 0$ such as our algorithms do not fail, we have:
\[\log C(N) = \sqrt[3]{\frac{64}9}(\log N)^{1/3}(\logTwo N)^{2/3}\left(\mathbf
A^{(n)}\left(
\dfrac{\logThree N}{\logTwo N}, \dfrac{1}{\logTwo
N}\right)+o\left(\dfrac{1}{(\logTwo N)^n}\right)\right).\]

Here are the first coefficients of the series $\mathbf A(X,Y) = \sum_{i,j\geq
0} a_{ij} X^i Y^j$ obtained via this implementation:

\[\begin{array}{c|l}
  a_{00} & 1\\
  a_{10} & 4/3\\
  a_{01} & -2\log 2+\log 3/6-2\\
  a_{20} & -4/9\\
  a_{11} & 4\log 2/3 - \log 3/9 + 4\\
  a_{02} & - (\log 2)^2 + \log 2\log 3/6 - 7(\log 3)^2/36- 6\log 2  + \log 3/2 - 5\\
  a_{30} & 32/81\\
  a_{21} & - 16\log 2/9 + 4\log 3/27 - 56/9\\
  a_{12} &8(\log 2)^2/3 - 4\log 2\log 3/9 + 56\log 2/3 + 14(\log 3)^2/27 -
  14\log 3/9 + 64/3\\
  a_{03} & - 4(\log 2)^3/3 + (\log 2)^2\log 3/3 - 14(\log 2)^2 - 
      7\log 2(\log 3)^2/9 \\
          & + 7\log 2\log 3/3 - 32\log 2 + 41
          (\log 3)^3/648 -
      49(\log 3)^2/18 + 8\log 3/3 -
      85/3
\end{array}\]

The
function $\xi$ in the introduction can be approximated
asymptotically by evaluating at $(\log_3 N/\log_2 N, 1/\log_2 N)$ the truncations of the series $\mathbf A(X,Y)-1$.
Using the algorithms in Section~\ref{sec:coeffs_expansion}, we were able to compute the series $\mathbf{A},
\mathbf{B}=\mathbf A$, and $\mathbf{D}$  
up to degree $14$ (more than a hundred terms). The fact that $\mathbf{A} = \mathbf{B}$ has been verified so far backs the claim that the patterns encountered while proving minimality are always as expected. Moreover, despite the fact the algorithms $\textsc{ProveExistence}$ and $\textsc{ProveMinimality}$ regularly have to consider terms $\cX^i \cY^j$ with $i$ or $j$ in $\frac{1}{2}\mathbb{Z}$ in the expansions of $a,b,d$, the coefficients of these terms always turned out to be zero in our experiments. These remarks allow us to formulate the following conjecture:

\begin{conjecture}\label{conj:general_shape}
  The minimizers $a,b,d$ of Problem~\ref{pb:optim2} belong respectively to the classes of
  functions $\cC^{[1/3,2/3]}$, $\cC^{[1/3,2/3]}$, $\cC^{[1/3,-1/3]}$. Moreover, the series
  $\mathbf A, \mathbf B\in\mathbb R[[X,Y]]$ associated to $a,b$ are equal.
\end{conjecture}

For
$i\geq 0$, we
let $\xi_i(N)$ denote the function $\mathbf A^{(i)}(\log_3 N/\log_2 N, 1/\log_2
N)-1$. In particular, for all $i\geq 0$, we have $\xi(N) = \xi_i(N) + o(1/(\log_2 N)^i)$. Figure
\ref{fig:zonecrypto} shows the behavior of $\xi_i$ for
cryptographically relevant values of $N$. Figure \ref{fig:convergence} focuses
on the range where we observe experimentally the convergence of the truncations $\xi_i$.

\begin{figure} 
\begin{center}
 \includegraphics[scale = 0.7]{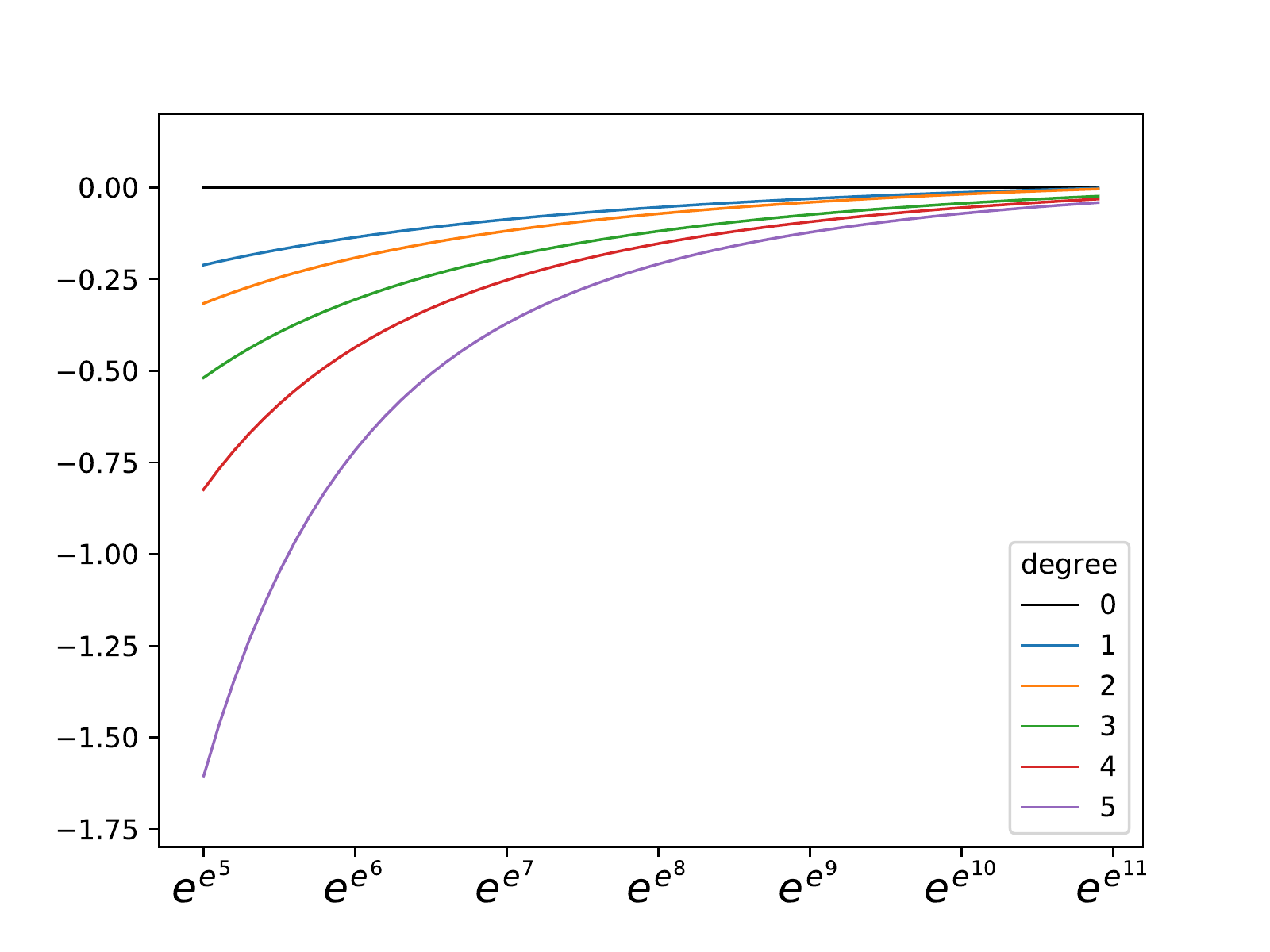}
  \caption{Truncations of $\xi$ up to total degree $i$ for $0\leq i\leq5$
  in function of $N$ for cryptographically relevant values of
  $N$. The abscissa axis is in $\log\relax\log$ scale.}
  \label{fig:zonecrypto}
\end{center}
\end{figure}
\begin{figure} 
\begin{center}
\includegraphics[scale = 0.7]{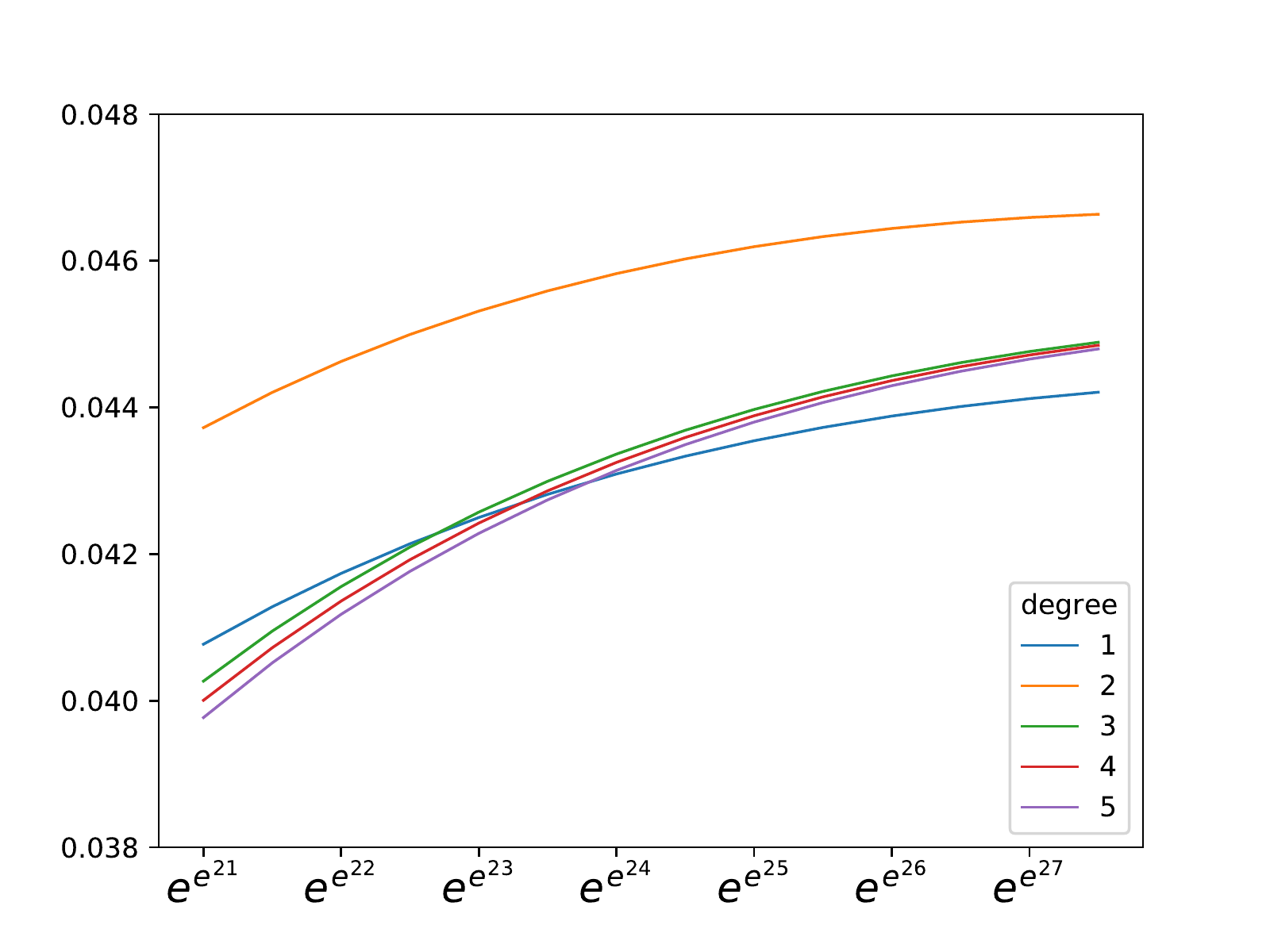}
  \caption{Converging behavior for $\xi_i$. The abscissa axis is in $\log\relax\log$
  scale.}
  \label{fig:convergence}
\end{center}
\end{figure}

These figures raise questions on the relevance of the traditional assumption
$\xi = 0$ for estimating the complexity of NFS in the range which is useful for
cryptographic applications, \emph{i.e.}, $N\leq 2^{20000}$. Indeed, we only start to
observe convergence for $N > \mathrm{exp}(\mathrm{exp}(25))\approx
2^{103881111194}$. Let us also notice that the convergence of $\xi$ to zero is
very slow as $N$ grows, since $\xi(N)\sim 4\log_3 N/(3\log_2 N)$
(Theorem~\ref{th:first_terms}). One could think that adding more and more terms
in the developement of $\xi$ would yield a more precise formula for the
complexity of NFS. However, it turns out that for practical values of $N$, replacing
$\xi$ by $\xi_{i}$ for $i>0$ is possibly even worse since the asymptotic series expansion
of $\xi$ seems to diverge for $N\leq \mathrm{exp}(\mathrm{exp}(25))$.  In
summary, all the asymptotic estimations for $\xi$ that we have at our disposal, including the brutal approximation $\xi = 0$, say
little to nothing about the behavior of the complexity of NFS in the 
range where the algorithm can be used. 

\smallskip

In fact, the expansion of $\xi$ relies on the expansion of $\rho$, and
the latter involves a series that converges only for sufficiently large
values as stated in Proposition~\ref{prop:radius}. We recall that $\rho(u) \sim \exp\left( -u \log u \left( \bfQ^{(i)}(\cX(u), \cY(u))
+ o \left(    \cY(u)^{i}  \right)      \right)  \right)$ where $\bfQ$ is
defined in Proposition \ref{prop:dev_integral_s}. To experimentally assess the
convergence properties of $\rho$ estimations, Figure~\ref{fig:logrho} plots the functions  $u \mapsto \bfQ
^{(i)}(\mathcal{X}(u), \mathcal{Y}(u))$ in function of $u$ for $1\leq i
\leq6$.

\begin{figure}
\begin{center}
\includegraphics[scale = 0.7]{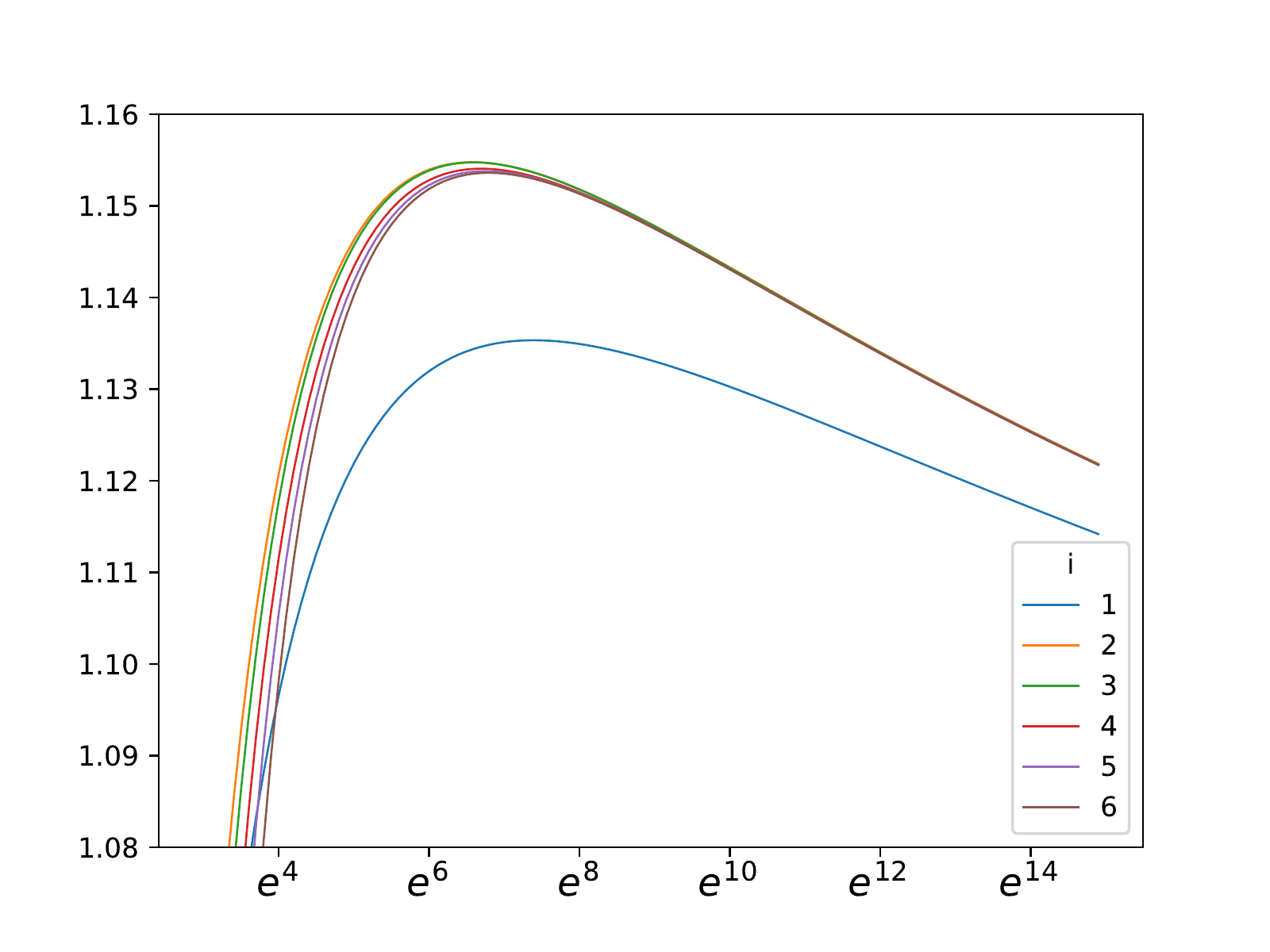}  
  \caption{Plot of the functions $ u \mapsto \bfQ ^{(i)}(\mathcal X(u),
  \mathcal Y(u))$ for $1\leq i \leq 6$, see
    Corollary~\ref{coro:dev_rho_only_rho}.}
  \label{fig:logrho}
\end{center}
\end{figure}

Experimentally, we observe in Figure~\ref{fig:logrho} that the asymptotic series expansion of $\rho$
starts to converge around $u\approx e^8$. When assessing the complexity of NFS, we evaluate
$\rho$ for values of $u$ that have the same order of magnitude than $(\log
N)^{1/3}$. This is consistent with the observed
convergence of the series expansion of $\xi$ for $N >
\mathrm{exp}(\mathrm{exp}(25))$ in Figure~\ref{fig:convergence} since $\exp(25)^{1/3}\approx e^8$.

\medskip

\textbf{Conclusion.} Under a few classical hypotheses and heuristics, we proved
that the function $o(1)$ hidden in the complexity of NFS decreases as $4\log_3
N/(3\log_2 N)$. We have also proposed an algorithm to compute an
asymptotic expansion of the function $\xi$ on which the
complexity of NFS rests. Unfortunately, replacing $\xi$ by a truncation of its
asymptotic expansion, even up to a high degree, in an attempt to have a better
understanding of the NFS complexity may be irrelevant for practical uses. Indeed,
it comes down to replacing a series by its first terms in a range where the
series diverges. Consequently, we recommend prudence when using
Formula~\eqref{eq:nfs-cplx-1} or other truncated asymptotic
expansions of the heuristic complexity of NFS in order to extrapolate keysizes
for cryptography. This stresses the importance of simulation tools
that rely on precise numerical
evaluations of $\rho$, or possibly on actual smoothness tests, to estimate the complexity of NFS.

\printbibliography 

\end{document}